\documentclass[11pt]{article}
\usepackage[T1]{fontenc}
\usepackage{amsfonts,amsmath,amsthm,amssymb,dsfont,mathtools}
\usepackage[margin=1in]{geometry}
\usepackage{fullpage}
\usepackage{enumerate}
\usepackage[linesnumbered,boxed,ruled,vlined]{algorithm2e}
\usepackage{thm-restate}
\usepackage[colorlinks,citecolor=blue,linkcolor=blue,urlcolor=red,pagebackref]{hyperref}
\usepackage[capitalise]{cleveref}
\usepackage{url}
\usepackage{graphicx}
\usepackage{mathdots}
\usepackage{paralist}
\usepackage[normalem]{ulem}
\usepackage{xspace}
\xspaceaddexceptions{]\}}
\usepackage{comment}
\usepackage{tabu}
\usepackage{framed}
\usepackage{float,wrapfig}
\usepackage{subfigure}
\usepackage{tikz}
\usepackage[usenames,dvipsnames]{pstricks}

\usepackage[mathlines]{lineno}
\usepackage{etoolbox}
\newcommand*\linenomathpatch[1]{%
  \cspreto{#1}{\linenomath}%
  \cspreto{#1*}{\linenomath}%
  \csappto{end#1}{\endlinenomath}%
  \csappto{end#1*}{\endlinenomath}%
}
\linenomathpatch{equation}
\linenomathpatch{gather}
\linenomathpatch{multline}
\linenomathpatch{align}
\linenomathpatch{alignat}
\linenomathpatch{flalign}

\theoremstyle{plain}
\newtheorem{theorem}{Theorem}[section]
\newtheorem{lemma}[theorem]{Lemma}

\newtheorem{corollary}[theorem]{Corollary}
\newtheorem{claim}[theorem]{Claim}

\newtheorem{prob}{Problem}

\theoremstyle{definition}
\newtheorem{definition}[theorem]{Definition}
\newtheorem{remark}[theorem]{Remark}

\def\ShowAuthNotes{1}
\ifnum\ShowAuthNotes=1
\newcommand{\authnote}[2]{\ \\ \textcolor{red}{\parbox{0.9\linewidth}{[{\footnotesize {\bf #1:} { {#2}}}]}}\newline}
\else
\newcommand{\authnote}[2]{}
\fi

\newcommand{\eps}{\varepsilon}

\newcommand{\Ex}{\operatorname*{\mathbf{E}}}

\newcommand{\poly}{\operatorname{\mathrm{poly}}}
\newcommand{\rank}{\operatorname{\mathrm{rank}}}
\newcommand{\polylog}{\poly\log}

\newcommand{\R}{\mathbb{R}}

\newcommand{\Z}{\mathbb{Z}}

\renewcommand{\tilde}{\widetilde}
\newcommand{\bod}{\boldsymbol}
\renewcommand{\vec}{\bod}

\newcommand{\caI}{\mathcal{I}}

\newcommand{\caP}{\mathcal{P}}

\newcommand{\caS}{\mathcal{S}}

\newcommand{\caW}{\mathcal{W}}

\makeatletter
\let\save@mathaccent\mathaccent
\newcommand*\if@single[3]{%
  \setbox0\hbox{${\mathaccent"0362{#1}}^H$}%
  \setbox2\hbox{${\mathaccent"0362{\kern0pt#1}}^H$}%
  \ifdim\ht0=\ht2 #3\else #2\fi
  }
\newcommand*\rel@kern[1]{\kern#1\dimexpr\macc@kerna}
\newcommand*\widebar[1]{\@ifnextchar^{{\wide@bar{#1}{0}}}{\wide@bar{#1}{1}}}
\newcommand*\wide@bar[2]{\if@single{#1}{\wide@bar@{#1}{#2}{1}}{\wide@bar@{#1}{#2}{2}}}
\newcommand*\wide@bar@[3]{%
  \begingroup
  \def\mathaccent##1##2{%
    \let\mathaccent\save@mathaccent
    \if#32 \let\macc@nucleus\first@char \fi
    \setbox\z@\hbox{$\macc@style{\macc@nucleus}_{}$}%
    \setbox\tw@\hbox{$\macc@style{\macc@nucleus}{}_{}$}%
    \dimen@\wd\tw@
    \advance\dimen@-\wd\z@
    \divide\dimen@ 3
    \@tempdima\wd\tw@
    \advance\@tempdima-\scriptspace
    \divide\@tempdima 10
    \advance\dimen@-\@tempdima
    \ifdim\dimen@>\z@ \dimen@0pt\fi
    \rel@kern{0.6}\kern-\dimen@
    \if#31
      \overline{\rel@kern{-0.6}\kern\dimen@\macc@nucleus\rel@kern{0.4}\kern\dimen@}%
      \advance\dimen@0.4\dimexpr\macc@kerna
      \let\final@kern#2%
      \ifdim\dimen@<\z@ \let\final@kern1\fi
      \if\final@kern1 \kern-\dimen@\fi
    \else
      \overline{\rel@kern{-0.6}\kern\dimen@#1}%
    \fi
  }%
  \macc@depth\@ne
  \let\math@bgroup\@empty \let\math@egroup\macc@set@skewchar
  \mathsurround\z@ \frozen@everymath{\mathgroup\macc@group\relax}%
  \macc@set@skewchar\relax
  \let\mathaccentV\macc@nested@a
  \if#31
    \macc@nested@a\relax111{#1}%
  \else
    \def\gobble@till@marker##1\endmarker{}%
    \futurelet\first@char\gobble@till@marker#1\endmarker
    \ifcat\noexpand\first@char A\else
      \def\first@char{}%
    \fi
    \macc@nested@a\relax111{\first@char}%
  \fi
  \endgroup
}
\makeatother

\newcommand{\ww}{w_{\max}}
\newcommand{\pp}{p_{\max}}

\newcommand{\dd}{\mathinner{.\,.\allowbreak}}
	
\newcommand{\supp}{\operatorname{\mathrm{supp}}}
\newcommand{\wts}{\operatorname{\mathrm{weights}}}

\newcommand{\sarkozy}{S\'{a}rk\"{o}zy\xspace}

\title{0-1 Knapsack in Nearly Quadratic Time}
\author{Ce Jin\thanks{cejin@mit.edu. Supported by NSF CCF-2129139, CCF-2127597, and a Siebel Scholarship.}\\MIT}
\date{\vspace{-1cm}}

\begin{document}

	\setcounter{page}{0} \clearpage
	\maketitle
	\thispagestyle{empty}
	\begin{abstract}
		We study pseudo-polynomial time algorithms for the fundamental \emph{0-1 Knapsack} problem.
		Recent research interest has focused on its fine-grained complexity with respect to the number of items $n$ and the \emph{maximum item weight} $w_{\max}$.
		Under $(\min,+)$-convolution hypothesis, $0$-$1$ Knapsack does not have $O((n+w_{\max})^{2-\delta})$ time algorithms (Cygan-Mucha-W\k{e}grzycki-W\l{}odarczyk 2017 and K\"{u}nnemann-Paturi-Schneider 2017). 
		 On the upper bound side,
     currently the fastest algorithm runs in  $\tilde O(n + \ww^{12/5})$ time (Chen, Lian, Mao, and Zhang 2023), improving the earlier $O(n + w_{\max}^3)$-time algorithm by Polak, Rohwedder, and W\k{e}grzycki (2021).

	 In this paper, we close this gap between the upper bound and the conditional lower bound (up to subpolynomial factors):
		\begin{itemize}
			\item The 0-1 Knapsack problem has a deterministic algorithm in $O(n + w_{\max}^{2}\log^4w_{\max})$ time.  
		\end{itemize}

		Our algorithm combines and extends several recent structural results and algorithmic techniques from the literature on  knapsack-type problems:
    \begin{enumerate}
      \item 
		We generalize the ``fine-grained proximity'' technique of Chen, Lian, Mao, and Zhang (2023) derived from the additive-combinatorial results of Bringmann and Wellnitz (2021) on dense subset sums. 
     This allows us to bound the support size of the useful partial solutions in the dynamic program.
\item 
To exploit the small support size,		our main technical component  is
		a vast extension of the ``witness propagation'' method, originally designed by Deng, Mao, and Zhong (2023) for speeding up dynamic programming in the easier unbounded knapsack settings. To extend this approach to our $0$-$1$ setting, we use a novel pruning method, as well as the two-level color-coding of Bringmann (2017) and the SMAWK algorithm on tall matrices.
    \end{enumerate}
	\end{abstract}

	\newpage

\maketitle

\setcounter{page}{0} \clearpage
\tableofcontents{}
	\thispagestyle{empty}
\newpage

\section{Introduction}

In the \emph{0-1 Knapsack} problem, we are given a knapsack capacity $t\in \Z^+$ and $n$ items $(w_1,p_1)$, $\dots$, $(w_n,p_n)$, where $w_i,p_i \in \Z^+$ denote the \emph{weight} and \emph{profit} of the $i$-th item, and we want to select a subset $X\subseteq [n]$ of items satisfying the capacity constraint $W(X):=\sum_{i\in X}w_i \le t$, while maximizing the total profit $P(X):=\sum_{i\in X}p_i$.

Knapsack is a fundamental problem in computer science.\footnote{In this paper we use the term Knapsack to refer to $0$-$1$ Knapsack (as opposed to other variants such as Unbounded Knapsack and Bounded Knapsack).}  It is among Karp's 21 NP-complete problems \cite{karp1972reducibility}, and the fastest known algorithm runs in $O(2^{n/2}n)$ time \cite{horowitz1974computing,schroeppel1981t}.
However, when the input integers are small, it is more preferable to use \emph{pseudopolynomial time} algorithms that have polynomial time dependence on both $n$ and the input integers.
Our work focuses on this pseudopolynomial regime. A well-known example of  pseudopolynomial algorithms is the textbook $O(nt)$-time Dynamic Programming (DP) algorithm for Knapsack, given by Bellmann \cite{Bellman1957} in 1957. 
Finding faster pseudopolynomial algorithms for Knapsack became an important topic in combinatorial optimization and operation research;
 see the book of Kellerer, Pferschy, and Pisinger \cite{0010031} for a nice summary of the results known by the beginning of this century.
In the last few years, research on Knapsack (and the easier Subset Sum problem,  which is the special case of Knapsack where $p_i=w_i$) has been revived by recent developments in fine-grained complexity (e.g, \cite{CyganMWW19,KunnemannPS17,KoiliarisX19,Bringmann17,BateniHSS18,AbboudBHS22}) and integer programming (e.g., \cite{EisenbrandW20,icalp21}), and the central question is to understand the best possible time complexities for solving these knapsack-type problems.

Cygan, Mucha, W\k{e}grzycki, and W\l{}odarczyk
\cite{CyganMWW19} and
  K\"{u}nnemann, Paturi, and Schneider 
\cite{KunnemannPS17} showed that the $O(nt)$ time complexity for Knapsack is essentially optimal
 (in the regime of $t=\Theta(n)$)
under the $(\min, +)$-convolution hypothesis.
To cope with this hardness result,
recent interest has focused on parameterizing the running time in terms of $n$ and the \emph{maximum item weight} $\ww$ (or the \emph{maximum item profit} $p_{\max}$), instead of the knapsack capacity $t$.
This would be useful when the item weights are much smaller than the capacity, and results along this line would offer us a more fine-grained understanding of knapsack-type problems. 
This parameterization  is also natural from the perspective of integer linear programming (e.g., \cite{EisenbrandW20}): when formulating Knapsack as an integer linear program, the maximum item weight $\ww$ corresponds to the standard parameter $\Delta$, maximum absolute value in the input matrix.

However, despite extensive research on $0$-$1$ Knapsack along these lines, our understanding about the dependence on $\ww$ is still incomplete. Known fine-grained lower bounds only ruled out $(n+\ww)^{2-\delta}$ algorithms for Knapsack 
\cite{CyganMWW19,KunnemannPS17} (for $\delta >0$).
In comparison,
Bellman's dynamic programming algorithm only runs in $O(nt)\le O(n^2 \ww)$ time.
Several papers obtained the bound $\tilde O(n\ww^2)$  via various methods \cite{EisenbrandW20,BateniHSS18,
AxiotisT19,KellererP04}.\footnote{We use $\tilde O(f)$ to denote $O(f\polylog f)$.}
Polak, Rohwedder, and W\k{e}grzycki
\cite{icalp21} carefully combined the \emph{proximity technique} of
Eisenbrand and Weismantel \cite{EisenbrandW20}
from integer programming with
the concave $(\max,+)$-convolution algorithm (\cite{KellererP04} or  \cite{smawk}), and obtained an $O(n+\ww^3)$ algorithm for Knapsack. 
These algorithms have cubic dependence on $(n+\ww)$.  
Finally, the very recent work by Chen, Lian, Mao, and Zhang \cite{chen2023faster} broke this cubic barrier  with an $\tilde O(n+\ww^{12/5})$-time algorithm, which was based on additive-combinatorial results of Bringmann and Wellnitz \cite{BringmannW21}.\footnote{An earlier work by Bringmann and Cassis \cite{esa} obtained an algorithm in $\tilde O(n\ww\pp^{2/3})$ time, which was the first algorithm for $0$-$1$ Knapsack with subcubic dependence on $(n+ \ww + \pp)$.} 

None of the above algorithms match the $(n+\ww)^{2-o(1)}$ conditional lower bound.
The following question has been asked by 
  \cite{icalp21,BringmannC22,chen2023faster}:
  \begin{center}
\textit{Main question: Can 0-1 Knapsack be solved in $\tilde O(n+\ww^2)$ time? }
\end{center}

 We remark that this $\tilde O(n+\ww^2)$ running time is known to be achievable for the easier \emph{Unbounded Knapsack} problem (where each item has infinitely many copies available) 
 \cite{AxiotisT19,ChanH22,dmz23}, matching the $(n+\ww)^{2-\delta}$ conditional lower bound for Unbounded Knapsack \cite{CyganMWW19,KunnemannPS17}.
 As argued by \cite{icalp21}, the 0-1 setting appears to be much more difficult, and most of the techniques for Unbounded Knapsack do not appear to apply to the 0-1 setting.

 \subsection{Our contribution}

In this paper, we affirmatively resolve this main question, closing the gap between the previous $\tilde O(n+\ww^{12/5})$ upper bound \cite{chen2023faster} and the quadratic conditional lower bound \cite{CyganMWW19,KunnemannPS17}.

\begin{theorem}
	\label{thm:knapsack-main}
The 0-1 Knapsack problem can be solved by a deterministic algorithm with time complexity  $O(  n + \ww^{2}\log^4 \ww)$.
\end{theorem}

In our paper we only describe an algorithm that outputs the total profit of the optimal knapsack solution. It can be modified to output an actual solution using the standard technique of back-pointers, without affecting the asymptotic time complexity.

By a reduction described in \cite[Section 4]{icalp21}, we have the following corollary which parameterizes the running time by the largest item profit $\pp$ instead of $\ww$.
\begin{corollary}
	\label{cor:knapsack-pmax}
The 0-1 Knapsack problem can be solved by a deterministic algorithm with time complexity  $O(  n + \pp^{2}\log^4\pp)$.
\end{corollary}

\paragraph*{Independent works.}
Independently and concurrently to our work, Bringmann \cite{karlnew} also obtained an $\tilde O(n+\ww^2)$ time algorithm for $0$-$1$ Knapsack (more generally, Bounded Knapsack).

\paragraph*{Chronological remarks.} 
The current paper is a substantially updated version of an earlier manuscript 
(posted to arXiv in July 2023).
This earlier manuscript contained much weaker results, and is obsolete now.
Our current paper incorporates part of the techniques from our earlier manuscript, and also builds on the very recent work by Chen, Lian, Mao, and Zhang \cite{chen2023faster} (posted to arXiv in July 2023).

\subsection{Technical overview}
Our Knapsack algorithm combines and extends several recent structural results and algorithmic techniques from the literature on  knapsack-type problems. 
In particular, we crucially build on the techniques from two previous papers by Chen, Lian, Mao, and Zhang \cite{chen2023faster}, and by Deng, Mao, and Zhong \cite{dmz23}. Now we review the techniques in prior works and describe the ideas behind our improvement.

\paragraph*{Fine-grained proximity based on additive combinatorics.}

There was a long line of work in the 80's and 90's on designing Subset Sum algorithms using techniques from \emph{additive combinatorics}
\cite{GalilM91, cfgalgo,chaimovichalgo,freimanalgo,freimanparti,chaimovichsurvey}, and more recently 
these  techniques have been revived and  applied to not only Subset Sum 
\cite{KoiliarisX19,MuchaW019,
BringmannW21,icalp21}
but also the more difficult Knapsack problem \cite{soda2023knapsack,chen2023faster}.
 Ultimately, these algorithms directly or indirectly rely on the following powerful result in additive combinatorics, pioneered by Freiman \cite{freiman93} and \sarkozy \cite{Sarkozy2} and tightened by Szemer\'{e}di and Vu \cite{szemeredivu}, and more recently strengthened by Conlon, Fox, and Pham \cite{fox}:
Let $\caS(A) = \{\sum_{b\in B} b: B\subseteq A\}$ denote the subset sums of $A$. Then, if set $A \subseteq  [N]$ has size $|A| \gg \sqrt{N} $, then $\caS(A)$ contains an arithmetic progression  of length $N$ (and this arithmetic progression is homogeneous, meaning that each element is an integer multiple of the common difference).

Another technique used in recent knapsack algorithms is the \emph{proximity technique} from the integer programming literature, see e.g., \cite{CookGST86,EisenbrandW20}.
When specialized to the Knapsack case (1-dimensional integer linear program), a proximity result
refers to a distance upper bound between the optimal knapsack solution and the \emph{greedy solution} (sort items in decreasing order of efficiencies $p_i/w_i$, and take the maximal prefix without violating the capacity constraint).
Polak, Rohwedder, and W\k{e}grzycki
\cite{icalp21} exploited the fact that these two solutions differ by at most $O(\ww)$ items,
which allowed them to shrink the size of the dynamic programming (DP) table from $t$ down to $O(\ww^2)$ (by performing DP on top of the greedy solution to find an optimal exchange solution). They achieved $O(n+\ww^3)$ time by batch-updating items of the same weight $w$ using the SMAWK algorithm \cite{smawk} (see also \cite{KellererP04,AxiotisT19}).

The very recent paper by Chen, Lian, Mao, and Zhang \cite{chen2023faster} developed a new ``fine-grained proximity'' technique that combines these two lines of approach. They used the 
additive-combinatorial results of Bringmann and Wellnitz \cite{BringmannW21}  
(which built on works of \sarkozy \cite{Sarkozy1,Sarkozy2} and Galil and Margalit \cite{GalilM91})
 to obtain several powerful structural lemmas involving the support size of two multisets $A,B$ (with integers from $[\ww]$) that avoid non-zero common subset sums, and these structural lemmas were translated into  proximity results using exchange arguments. These fine-grained proximity results of \cite{chen2023faster} are more powerful than the earlier proximity bounds used in \cite{icalp21,EisenbrandW20}; the following lemma from \cite{chen2023faster} is one example: given a Knapsack instance, we can partition the item weights  into two subsets $[\ww]=\caW_1\uplus \caW_2$, such that $|\caW_1|\le \tilde O(\sqrt{\ww})$, and the differing items  between the greedy solution and the optimal solution whose weights belong to $\caW_2$ can only have total weight $O(\ww^{3/2})$.
This lemma immediately led to a simple $\tilde O(n + \ww^{5/2})$ algorithm \cite{chen2023faster}.
A bottleneck step in this algorithm is to use DP to compute partial solutions consisting of items with weights from $\caW_1$:  they need to perform the batch DP update (based on SMAWK) $|\caW_1|$ times, and the size of the DP table is still $O(\ww^2)$ as in \cite{icalp21}, so the total time for this step is $\tilde O(\ww^{2.5})$.
To overcome this bottleneck, \cite{chen2023faster} used more refined proximity results based on the multiplicity of item weights, and obtained an improved running time $\tilde O(n + \ww^{2.4})$.

\paragraph*{DP strategy based on multiplicity.}
In our work, we completely overcome this bottleneck of \cite{chen2023faster}: we can implement the DP for items with weights from $\caW_1$ in only $\tilde O(\ww^2)$ time.
This is the main technical part of our paper.
(The other bottleneck in \cite{chen2023faster}'s simple $\tilde O(n+\ww^{5/2})$ time algorithm is to deal with items whose weights come from $\caW_2$, but this part can be improved more easily by dividing into $O(\log \ww)$ partitions with smoothly changing parameters. See \cref{subsec:chenpartition}.)
Now we give an overview of our improvement.

We rely on another additive-combinatorial lemma (\cref{lem:exchange-l0proximity}) which can be derived from the results of Bringmann and Wellnitz \cite{BringmannW21}; it is analogous and inspired by the fine-grained proximity results of \cite{chen2023faster}, but is not directly comparable to theirs.  It implies the following proximity result: Let $D$ denote the set of differing items between the greedy solution and the optimal solution. Then, for any $r\ge 1$, there can be at most $\tilde O(\sqrt{\ww/r})$ many weights $w\in [\ww]$ such that $D$ contains at least $r$  items of weight $w$ (i.e., $w$ has multiplicity $\ge r$ in the item weights of $D$). In other words, if we figuratively think of the histogram of the weights of items in $D$, then the number of columns in the histogram with height $\ge r$ should be at most $\tilde O(\sqrt{\ww/r})$.
As a corollary, the total area below height $r$ in this histogram is at most $\sum_{r'=1}^r \tilde O(\sqrt{\ww/r'}) = \tilde O(\sqrt{r \ww})$.

Our DP algorithm exploits the aforementioned structure of $D$ as follows. 
We perform the DP in $O(\log \ww)$ phases, where in the $j$-th phase $(j\ge 1)$ we update the current DP table with all items of \emph{rank} in $[2^{j-1},2^j)$. Here, the \emph{rank} of a weight-$w$ item is defined as the rank of its profit among all weight-$w$ items (an item with rank $1$ is the most profitable item among its weight class). 
By the end of phase $j$, our DP table should contain the partial solution consisting of all items in $D$ of rank $<2^j$, i.e., the partial solution that corresponds to the part below height $2^j$ in the histogram representing $D$.
As we mentioned earlier, this partial solution only has $\tilde O(\sqrt{2^j \ww})$ items, and hence $\tilde O(\ww\cdot \sqrt{2^j \ww})$ total weight, so the size of the DP table at the end of phase $j$ only needs to be $L_j := \tilde O(\ww\cdot \sqrt{2^j \ww})$.

 To efficiently implement the DP in each phase, we need to crucially exploit the aforementioned fact that the number of weights $w\in \caW_1$ with multiplicity $\ge 2^{j-1}-1$ in $D$ is at most $b_j:= \tilde O(\sqrt{\ww/2^j})$. (Note that in phase $j=1$ this threshold is $2^{j-1}-1=0$, and the upper bound $b_1= \tilde O(\sqrt{\ww})$ simply follows from $|\caW_1|= \tilde O(\sqrt{\ww})$ guaranteed by \cite{chen2023faster}'s partition.)  
 Our goal is to perform each phase of the DP updates in $\tilde O(b_j\cdot L_j) = \tilde O(\ww^2)$ time.
  To achieve this goal, we surprisingly adapt a recent technique introduced in the much easier \emph{unbounded} knapsack settings by Deng, Mao, and Zhong \cite{dmz23}, called ``witness propagation''. In the following we briefly review this technique.

\paragraph*{Transfer of techniques from the unbounded setting.}
The \emph{unbounded} knapsack/subset sum problems, where each item has infinitely many copies available, are usually easier for two main reasons:
	(1)  Since there are infinite supply of items, we do not need to keep track of which items are used so far in the DP.
	(2)  There are more powerful structural results available, in particular the Carath{\'{e}}odory-type theorems \cite{EisenbrandS06,Klein22,ChanH22,dmz23} which show the existence of optimal solution vectors with only logarithmic support size. 

  Deng, Mao, Zhong \cite{dmz23} recently exploited the small support size to design near-optimal algorithms for several unbounded-knapsack-type problems, based on their key new technique termed ``witness propagation''.
    The idea is that, since the optimal solutions must have small support size (but possibly with high multiplicity), one can first prepare the ``base solutions'', which are partial solutions with small support and multiplicity at most one. Then, they gradually build full solutions from these base solutions, by ``propagating the witnesses'' (that is, increase the multiplicity of some item with non-zero multiplicity). The time complexity of this approach is low since the support sizes are small.

	Now we come back to our DP framework for $0$-$1$ knapsack described earlier, and observe that we are in a very similar situation to the unbounded knapsack setting of \cite{dmz23}. %
	In our case, if we intuitively view our DP as gradually growing the columns of the histogram representing $D$, then after phase $j-1$, there can be only $\le b_j$ columns in the histogram that may continue growing in subsequent phases. 
	This means the ``active support'' of our partial solutions has size $\le b_j$:  when we extend a partial solution in the DP table during phase $j$, we only need to consider items from $b_j$ many weight classes, namely those weights that have ``full multiplicity'' in this partial solution by the end of phase $j-1$. (If there are more than $b_j$ many such weights, then the proximity result implies that this partial solution cannot be extended to the optimal solution, and we can safely discard it.)
This gives us hope of implementing the DP of each phase in $\tilde O(b_j\cdot L_j) = \tilde O(\ww^2)$ using the witness propagation idea from \cite{dmz23}.

However, we still need to overcome several difficulties that arise from the huge difference between 0-1 setting and unbounded setting. In particular, the convenient property (1) for unbounded knapsack mentioned above no longer applies to the 0-1 setting.
In the following we briefly explain how we implement the witness propagation idea in the $0$-$1$ setting.

\paragraph*{Witness propagation in the $0$-$1$ setting.}
In each phase $j$ of our DP framework, we are faced with the following task (from now on we drop the subscript $j$ and denote $b=b_j, L=L_j$): we are given a DP table $q[\,]$ of size $L$, in which each entry $q[z]$ is associated with a set $S[z]\subseteq \caW_1$ of size $|S[z]|\le b$ (this is the ``active support'' of the partial solution corresponding to $q[z]$).
For each entry $q[z]$, we would like to extend this partial solution by adding items whose weights come from $S[z]$. More specifically, letting $x_w\ge 0$ denote the number of weight-$w$ items to add (where $w\in S[z]$), we should update the final DP table entry $q'[z+ \sum_{w\in S[z]}x_w w]$ with the new profit $q[z] + \sum_{w\in S[z]} Q_w(x_w)$. Here $Q_w(x)$ is the total profit of the top $x$ remaining items of weight $w$ (note that $Q_w(\cdot)$ is concave).
Our goal is to compute the final DP table $q'[\,]$ (which should capture the optimal ways to extend from $q[\,]$) in $\tilde O(bL)$ time.
(Note that in the idealistic setting where all $S[z]$ are contained in a common superset $\hat S$ of size $|\hat S|\le b$, this task can be solved via standard applications of SMAWK in $O(bL)$ total time in the same way as \cite{icalp21,AxiotisT19,KellererP04}. The key challenge in our setting is that, although each $S[z]$ has size $\le b$, their union over all $z$ may have much more than $b$ types of weights.)

We first focus on an interesting basic case where each set $S[z]$ has size at most $b=1$.  In this case, for each DP table entry $q[z]$ with $S[z] = \{w\}$ we would like to perform the DP update $q'[i] \gets \max(q'[i], q[z] + Q_w((i-z)/w))$ for all $i$ such that $i\ge z$ and $i \equiv z \pmod{w}$.
Similarly to \cite{icalp21,AxiotisT19,KellererP04}, we try to use the SMAWK algorithm to perform these DP updates. 
However, since these sets $S[z]$ may contain different types of weights $w$, we need to deal with them separately.
This means that for each weight $w$, there may be only sublinearly many indices $z$ with $S[z] = \{w\}$. 
Hence, in order to save time, we need to do SMAWK for each $w$ in time complexity sublinear in the entire DP table size $L$, and only near-linear in $n_w = \big \lvert \big \{z: S[z]= \{w\}\big \}\big \rvert$.
So we need to let SMAWK return a compact output representation, which partitions the DP table into $n_w$ segments, or more precisely, $n_w$ arithmetic progressions (APs) of difference $w$, where each $z \in \big \{z: S[z]= \{w\}\big \}$ is associated with an AP containing the indices $i$ for which $q'[i]$ is maximized by $q[z]+Q_w((i-z)/w)$.
 This is an very interesting scenario where we actually need to use the tall-matrix version of SMAWK.

Then, we need to update these APs returned by these SMAWK algorithm invocations (for various different weights $w$) to the DP table $q'[\,]$.
That is, for each $i$, we would like to pick the AP that contains $i$ and maximizes the profit $q[z]+Q_w((i-z)/w)$ mentioned earlier. 
Naively going through each element in every AP would take time proportional to the total length of these APs.
This would be too slow: although the total number of APs is only $O(L)$, their total length could still be very large.
  To solve this issue, we design a novel skipping technique, so that we can ignore suffixes of some of the APs, while still ensuring that we do not lose the optimal solution, so that the total time is reduced to $\tilde O(L)$. 
  
We first explain the key insight behind our skipping technique, through the following example. Suppose index $i$ is contained in two APs computed by SMAWK for two different weights $w_1\neq w_2$, denoted by $I_1 = \{z_1+x w_1: \ell_1\le x\le r_1\}$ and $I_2 =  \{z_2+x w_2: \ell_2\le x\le r_2\}$. The final DP table entry $q'[i]$ is updated using $ \max\{q[z_1] + Q_{w_1}((i-z_1)/w_1), q[z_2] + Q_{w_2}((i-z_2)/w_2)\}$, and we suppose the first option is larger. Then, we claim that all elements in $I_2 \cap (i,+\infty)$ are useless. To see this, consider any $i^*\in I_2\cap (i,+\infty)$, and denote $i^*=z_2+x^*w_2, i= z_2+xw_2$ ($x^*>x$), so $i^*\in I_2$ represents a solution of total weight $i^*$ and profit $q[z_2]+Q_{w_2}(x^*)$.
However, we can show this solution represented by $i^*\in I_2$ is dominated by another solution defined as follows: add $(x^*-x)$ many weight-$w_2$ items to the solution represented by $i\in I_1$, achieving the same total weight $i + (x^*-x)w_2 = i^*$ but higher (or equal) total profit $q[z_1] + Q_{w_1}((i-z_1)/w_1) + Q_{w_2}(x^*-x) \ge q[z_2] + Q_{w_2}(x) + Q_{w_2}(x^*-x) \ge q[z_2] + Q_{w_2}(x^*)$ (recall $Q_{w_2}(\cdot)$ is concave).
Hence, we can safely ignore the solution represented by $i^*\in I_2$ without affecting optimality.\footnote{We need more tie-breaking arguments to deal with the possibility that $i^*\in I_2$ is not strictly dominated (i.e., they have equal profit). We omit them in this overview.}$^{,}$\footnote{The solution we showed that dominates the solution represented by $i^*\in I_2$ is not represented by any AP element, as it uses items of both types of weights $w_1$ and $w_2$. It is possible that $i^*\in I_2$ is still the best weight-$i^*$ solution among those represented by the AP elements (which use only one type of weight), but it is fine to omit it since eventually it is not useful for the optimal knapsack solution.}

  The key insight above can be naturally used to design the following skipping technique:
  We initialize an empty bucket $B[i]$ for each index $i$ in the DP table. For each of the $O(L)$ many APs returned by SMAWK, we insert the (description of the) AP into the bucket indexed by the beginning element of this AP.  Then we iterate over the buckets $B[i]$ in increasing order of $i$. For each $B[i]$, we pick the AP from this bucket that maximizes the profit value at $i$, and update the profit value $q'[i]$ accordingly.  Then, we copy this maximizing AP from bucket $B[i]$ to the bucket indexed by the successor of $i$ in this AP; the other non-maximizing APs in bucket $B[i]$ will not be copied. In this way, the total time is $O(L)$, since we start with $O(L)$ APs and each bucket only copies one AP to another bucket.

Now we briefly describe how to generalize from the $|S[z]|\le 1$ case to $|S[z]|\le b$ for larger $b$. First we make an ideal assumption that we can partition all possible weights into $b$ parts,  $\caW_1 = \caW^{(1)}\uplus \caW^{(2)}\uplus \cdots \uplus \caW^{(b)}$, so that $|S[z] \cap \caW^{(k)}| \le 1$ for all $z$ and $k$. In this ideal case, we can iteratively perform $b$ rounds, where in the $k$-th round we restrict the sets $S[z]$ to $S[z] \cap \caW^{(k)}$, and perform the DP updates using the $b=1$ case algorithm described above in $\tilde O(L)$ time. (Note that after each round we should modify the active supports $S[z]$ accordingly: if $q'[i]$ is updated using $q[z]+ Q_w((i-z)/w)$ for some $w$ in this round, then the new $S[i]$ for the next round should be the old $S[z]$.) Hence the total time is $\tilde O(bL)$. 
In the non-ideal case, we use the two-level color-coding technique originally used by Bringmann \cite{Bringmann17} in his subset sum algorithm. This technique gives us some properties that are weaker than the ideal assumption but still allow us to apply basically the same idea as the ideal case.

The correctness of our algorithm described above (namely that our skipping technique does not lose the optimal knapsack solution) is intuitive and is based on exchange arguments, but it takes some notations and definitions  to formally write down the proof.
In the main text of the paper, we formalize the intuition above, and abstract out a core problem called $\textsc{HintedKnapsackExtend}^+$ (\cref{prob:prob3}) that captures the scenario described above in a more modular way, and  prove some helper lemmas for \cref{prob:prob3} (for example, to allow us to decompose an instance with large $b$ to multiple instances with smaller $b$).

\subsection{Further related works}
In contrast to our 0-1 setting, the \emph{unbounded} setting (where each item has infinitely many copies available) has also been widely studied in the literature of Knapsack and Subset Sum algorithms, e.g.,  \cite{Lincoln0W20,JansenR19,moor,AxiotisT19,ChanH22,Klein22,dmz23}.

For the easier Subset Sum problem, an early result for Subset Sum in terms of $n$ and $\ww$ is Pisinger's deterministic $O(n\ww)$-time algorithm for Subset Sum \cite{Pisinger99}. This is not completely subsumed by Bringmann's $\tilde O(n+t)\le \tilde O(n\ww)$ time algorithm \cite{Bringmann17}, due to the extra log factors and randomization in the latter result.
More recently, Polak, Rohwedder, and W\k{e}grzycki \cite{icalp21} observed that an $\tilde O(n+\ww^2)$ time algorithm directly follows from combining their proximity technique with Bringmann's $\tilde O(n+t)$ Subset Sum algorithm \cite{Bringmann17}.
 They improved it to $\tilde O(n+\ww^{5/3})$ time, by further incorporating additive combinatorial techniques by \cite{BringmannW21}. Very recently, \cite{chen2023faster} obtained $\tilde O(n+\ww^{3/2})$-time algorithm for Subset Sum, using their fine-grained proximity technique based on additive combinatorial results of \cite{BringmannW21}.

Recently there has also been a lot of work on approximation algorithms for Knapsack and Subset Sum (and Partition) \cite{Chan18a,Jin19,MuchaW019,BringmannN21,soda2023knapsack,xiaomaofptas,clmzfptas}. 
Prior to this work, the fastest known $(1-\eps)$ approximation algorithm for 0-1 Knapsack had time complexity $\tilde O(n+1/\eps^{2.2})$ \cite{soda2023knapsack}.
 Notably, \cite{soda2023knapsack} also used the additive combinatorial results of \cite{BringmannW21} to design knapsack approximation algorithms; this was the first application of additive combinatorial techniques to knapsack algorithms. 
In August 2023, Mao \cite{xiaomaofptas} and Chen, Lian, Mao, and Zhang \cite{clmzfptas} independently improved the time complexity to $\tilde O(n+1/\eps^{2})$, which is nearly tight under the $(\min,+)$-convolution hypothesis \cite{CyganMWW19,KunnemannPS17}.

\subsection{Open problems}

There are several interesting open questions.
\begin{itemize}
    \item In the regime where $n$ is much smaller than $\ww$, can we get faster algorithms for $0$-$1$ Knapsack? 
    The independent work of He and Xu \cite{hqz} achieved $\tilde O(n^{1.5}\ww)$ time. 
    By combining with our result, one can also bound the running time as $\tilde O(n + \min\{n^{1.5}\ww , \ww^2\}) \le \tilde O(n\ww^{4/3})$.
    Can we achieve $\tilde O(n\ww)$ time (which would also match the $(n+\ww)^{2-o(1)}$ conditional lower bound based on $(\min,+)$-convolution hypothesis~\cite{CyganMWW19,KunnemannPS17})?
    \item Can we solve $0$-$1$ Knapsack in $O((n+\ww+\pp)^{2-\delta})$ time for any constant $\delta>0$? Bringmann and Cassis \cite{BringmannC22}
    gave algorithms of such running time for the easier unbounded knapsack problem. They also showed that such algorithms require computing  bounded-difference $(\min,+)$-convolution \cite{ChanL15,ChiDX022}.
    \item Can we solve $0$-$1$ Knapsack in $O(n+\ww^2/2^{\Omega(\sqrt{\log \ww})})$ time, matching the best known running time for $(\min,+)$-convolution  \cite{Williams18,bremner2014necklaces,ChanW21}? Algorithms with such running time are known for the easier unbounded knapsack problem  \cite{AxiotisT19,ChanH22,dmz23}.
        \item Can Subset Sum be solved in $\tilde O(n+\ww)$ time? This question has been repeatedly asked in the literature \cite{AxiotisBJTW19,AbboudBHS22,BringmannW21,icalp21,BringmannC22}. Currently the best result is the very recent $\tilde O(n+\ww^{3/2})$-time randomized algorithm by Chen, Lian, Mao, and Zhang~\cite{chen2023faster}.
            \item 
			Can our techniques be useful for other related problems, such as scheduling \cite{BringmannFHSW22,AbboudBHS22jcss,kpr23} or low-dimensional integer linear proramming \cite{EisenbrandW20}?
\end{itemize}

\subsection*{Paper organization}
\cref{sec:prelim} contains definitions, notations, and some lemmas from previous works, which are essential for understanding \cref{sec:knapsack}.
Then, in \cref{sec:knapsack} we describe our algorithm for $0$-$1$ Knapsack. A key subroutine of our algorithm is deferred to \cref{sec:propagation}.

\section{Preliminaries}
\label{sec:prelim}

\subsection{Notations and definitions}
\label{subsec:generalnotation}
We use $\tilde O(f)$ to denote $O(f\polylog f)$.
Let $[N] = \{1,2,\dots,N\}$.
\paragraph*{Multisets and subset sums.}
For an integer multiset $X$, and an integer $x$, we use $\mu_X(x)$ to denote the multiplicity of $x$ in $X$.
For a multiset $X$, the \emph{support} of $X$ is the set of elements it contains, denoted as $\supp(X):=\{x: \mu_X(x)\ge 1\}$.
We say a multiset $X$ is \emph{supported on} $[N]$ if $\supp(X) \subseteq [N]$.
For multisets $A,B$ we say $A$ is a subset of $B$ (and write $A\subseteq B$) if for all $a\in A$, $\mu_B(a) \ge \mu_A(a)$.
We write $A \uplus B$ as the union of $A$ and $B$ by adding multiplicities.

 The \emph{size} of a multiset $X$ is $|X| = \sum_{x\in \Z}\mu_X(x) $, and the \emph{sum of elements} in $X$ is $\Sigma(X)= \sum_{x\in \Z}x\cdot \mu_X(x)$.
 The set of all \emph{subset sums} of $X$ is $\caS(X):= \{ \Sigma(Y): Y\subseteq X\} $. We also define $\caS^*(X):= \{ \Sigma(Y): Y\subseteq X, Y\neq \varnothing\} $ to be the set of subset sums formed by \emph{non-empty} subsets of $X$.

The $r$-\emph{support} of a multiset $X$ is the set of items in $X$ with multiplicity at least $r$, denoted as $\supp_r(X):= \{x: \mu_X(x)\ge r\}$.

\paragraph*{Vectors and arrays.}
We will work with vectors in $\Z^{\caI}$ where $\caI$ is some index set. We sometimes denote vectors in boldface, e.g., $\vec x\in \Z^\caI$, and use non-boldface with subscript to denote its coordinate, e.g., $x_i \in \Z$ (for $i\in \caI$).
Let $\supp(\vec x):= \{i\in \caI: x_i \neq 0\}$,
$\|\vec x\|_0 :=  \lvert  \supp(\vec x)  \rvert$,
and
$\|\vec x\|_1 := \sum_{i\in \caI}|x_i|$.
Let $\vec{0}$ denote the zero vector.
For $i\in \caI$, let $\vec e_i$ denote the unit vector with $i$-th coordinate being $1$ and the remaining coordinates being $0$.

We use $A[\ell\dd r]$ to denote an array indexed by integers $i\in \{\ell,\ell+1,\dots,r\}$. The $i$-th entry of the array is $A[i]$. Sometimes we consider arrays of vectors, denoted by ${\vec x}[\ell \dd r]$, in which every entry ${\vec x}[i] \in \Z^{\caI}$ is a vector, and we use $x[i]_j$ to denote the $j$-th coordinate of the vector ${\vec x}[i]$ (for $j\in \caI$).

\paragraph*{0-1 Knapsack.}
In the $0$-$1$ Knapsack problem with $n$ input items $(w_1,p_1),\dots,(w_n,p_n)$ (where \emph{weights} $w_i \le \ww$ and \emph{profits} $p_i\le \pp$ are positive integers) and knapsack capacity $t$, an \emph{optimal knapsack solution} is an item subset $X\subseteq [n]$ that maximizes the total profit
\begin{equation}
    P(X):=\sum_{i\in X}p_i,
\end{equation}
subject to the capacity constraint
\begin{equation}
    W(X):=\sum_{i\in X}w_i \le t.
\end{equation}

We will frequently use the following notations:
\begin{itemize}
\item 
Let $\caW = \supp(\{w_1,w_2,\dots,w_n\}) \subseteq [\ww]$ be the set of input item weights.
\item 
For $\caW'\subseteq \caW$, let $I_{\caW'}:= \{ i\in [n]: w_i\in \caW'\}$ denote the set of items with weights in $\caW'$.
    \item For $I = \{i_1,\dots,i_{|I|}\}\subseteq [n]$, 
let $\wts(I) = \{w_{i_1},\dots,w_{i_{|I|}}\}$ be the \emph{multiset} of weights of items in $I$.
\end{itemize}

We assume $\ww \le t$ by ignoring items that are too large to fit into the  knapsack.  We assume $w_1+\dots + w_n>t$, since otherwise the trivial optimal solution is to include all the items.
 We assume $\ww\le n^2$, because when $\ww>n^2$ it is faster to run the textbook dynamic programming algorithm \cite{Bellman1957} in $O(nt)\le O(n\cdot n\ww)\le O(\ww^2)$ time. 
 We use the standard word-RAM computation model with $\Theta(\log n)$-bit words, and we assume $p_i\le \pp$ fits into a single machine word.\footnote{If this assumption is dropped, we simply pay an extra $O(\log \pp)$ factor in the running time for adding integers of magnitude $(n\pp)^{O(1)}$.}

The \emph{efficiency} of item $i$ is $p_i/w_i$.
We always assume the input items have 
\emph{distinct} efficiencies $p_i/w_i$.   This assumption is justified by the following
tie-breaking lemma
proved in \cref{appendix:breakties}.
\begin{restatable}[Break ties]{lemma}{breakties}
    \label{lem:breakties}
    Given a 0-1 Knapsack instance $I$, in $O(n)$ time we can deterministically reduce it to another 0-1 Knapsack instance $I'$ with $n,\ww$ and $t$ unchanged, and $\pp' \le  \poly(\pp,\ww, n)$, such that the items in $I'$ have %
     \emph{distinct efficiencies} and distinct profits.
\end{restatable}

\subsection{Greedy solution and proximity}
\label{subsec:maximalprefix}
\paragraph*{Greedy solution.}
Sort the $n$ input items in decreasing order of efficiency,
\begin{equation}
    \label{eqn:sortstrict}
    p_1/w_1 > p_2/w_2 > \dots > p_n/w_n.
\end{equation}
The \emph{greedy solution} (or \emph{maximal prefix solution}) is the item subset
\begin{equation}
    \label{eqn:greedysol}
    G=\{1,2,\dots, i^*\}, \text{ where } i^* = \max\{ i^* : w_1+w_2+\dots+w_{i^*} \le t\},
\end{equation}
i.e.,  we greedily take the most efficient items one by one, until the next item cannot be added without exceeding the knapsack capacity.
Since the input instance is nontrivial, we have $1\le i^*\le n-1$, and $W(G) \in (t-\ww,t]$.
Denote the remaining items as
$\widebar{G}= [n]\setminus G = \{i^*+1,i^*+2,\dots,n\}$.

\begin{remark}
    As noted by 
    \cite{icalp21}, the greedy solution $G$ can be found in deterministic $O(n)$ time using linear-time median finding algorithms \cite{BlumFPRT73} (if we only need the set $G$ rather than the order of their elements), as opposed to a straightforward $O(n\log n)$-time sorting according to \cref{eqn:sortstrict}.
\end{remark}

Every item subset $X\subseteq [n]$ can be written as $X=(G\setminus B) \cup A$ where $A\subseteq \widebar{G}$ and $B\subseteq G$.
Finding an optimal knapsack solution $X$ is equivalent to finding an \emph{optimal exchange solution},  defined as a pair of subsets $(A,B)$ ($A\subseteq \widebar{G}, B\subseteq G$) that maximizes $P(A)-P(B)$ subject to $W(A)-W(B) \le t-W(G)$.
Since any optimal knapsack solution $X$ satisfies $W(X) \in (t-\ww,t]$, we have
\begin{equation}
    \label{eqn:sumclose}
    0\le W(A) - W(B) = W(X) - W(G) < \ww
\end{equation}
for any optimal exchange solution $(A,B)$.

\paragraph*{Proximity.}
For any optimal exchange solution $(A,B)$, a simple exchange argument shows that the weights of items in $A$ and in $B$ do not share any non-zero common subset sum, i.e.,
\begin{equation}
    \label{eqn:nocommon}
    \caS^*(\wts(A)) \cap \caS^*(\wts(B)) = \varnothing.
\end{equation}
Indeed, for an  optimal knapsack solution $X = (G\setminus B)\cup A$, if non-empty item sets $A'\subseteq A$ and $B'\subseteq B $ have the same total weight, then $(X\cup B')\setminus A'$ is a set of items with the same total weight as $X$ but \emph{strictly} higher total profit (since efficiencies of items in $B' \subseteq G$ are strictly higher than efficiencies of items in $A'\subseteq \widebar{G}$ due to \cref{eqn:sortstrict,eqn:greedysol}), contradicting the optimality of $X$.

The following proximity bound \cref{eqn:l1proximity} is consequence of \cref{eqn:sumclose} and \cref{eqn:nocommon}, and was used in previous works such as \cite{icalp21,chen2023faster} (see e.g., {\cite[Lemma 2.1]{icalp21}} for a short proof):
 for any optimal exchange solution $(A,B)$, it holds that
\begin{equation}
    \label{eqn:l1proximity}
    |A|+|B| \le 2\ww.
\end{equation}
In other words, any optimal knapsack solution $X$ differs from the greedy solution $G$ by at most $2\ww$ items. The bound of \cref{eqn:l1proximity} immediately implies 
\begin{equation}
    \label{eqn:l1proximitysum}
    W(A)+W(B) \le 2\ww^2
\end{equation}
for any optimal exchange solution $(A,B)$.

\paragraph*{Weight classes and ranks.}
We rank items of the same weight $w$ according to their profits, as follows:
\begin{definition}[Rank of items]
    \label{defn:rank}
For each $w\in \caW$, consider the weight-$w$ items outside the greedy solution, $\widebar G \cap I_{\{w\}} = \{i_1,i_2,\dots,i_m\}$, where $p_{i_1}>p_{i_2}>\dots > p_{i_m} $.
We define $\rank(i_1)=1,\rank(i_2)=2,\dots, \rank(i_m)=m$.
Similarly, consider the weight-$w$ items in the greedy solution,
 $G \cap I_{\{w\}} = \{i'_1,i'_2,\dots,i'_{m'}\}$, where $p_{i'_1}<p_{i'_2}<\dots < p_{i'_{m'}}$. 
We define $\rank(i'_1)=1,\rank(i'_2)=2,\dots, \rank(i'_{m'})=m'$.
In this way, every item $i\in [n]$ receives a $\rank(i)$.
\end{definition}

Then, a standard observation is that an optimal solution should always take a prefix from each weight class:
\begin{lemma}[Prefix property]
    \label{lem:prefix}
   Consider any optimal exchange solution $(A,B)$. If $i\in A$, then $\{i' \in \widebar{G} \cap I_{w_i} : \rank(i')\le \rank(i)\}\subseteq A$, and $\rank(i)\le 2\ww$.
   
   Similarly, if $i\in B$, then $\{i' \in {G} \cap I_{w_i} : \rank(i')\le \rank(i)\}\subseteq B$, and $\rank(i)\le 2\ww$.
\end{lemma}
\begin{proof}
    We only prove the $i\in A$ case. The $i\in B$ case is symmetric.  Consider two weight-$w$ items $i,i'\in \widebar{G} \cap I_{w}$ with $\rank(i')<\rank(i)$ and $i\in A$.
    Suppose for contradiction that 
    $i'\notin A$. Then, $p_{i'}>p_i$, and hence $((A\setminus \{i\})\cup \{i'\},B)$ is an exchange solution with the same weight as $(A,B)$ but achieving strictly higher profit, contradicting the optimality of $(A,B)$. Hence, $A$ contains all $i'\in \widebar{G} \cap I_{w}$ with $\rank(i')\le \rank(i)$.  Since $|A|\le 2\ww$ by \cref{eqn:l1proximity}, we must have $\rank(i)\le 2\ww$ for $i\in A$.
\end{proof}

We remark that all items $i\in [n]$ with $\rank(i)\le 2\ww$ can be deterministically selected and sorted in $O(n + \ww^2\log \ww)$ time using linear-time median selection algorithms \cite{BlumFPRT73}.

\subsection{Dynamic programming and partial solutions}
Our algorithm uses dynamic programming (DP) to find an optimal exchange solution $(A,B)$ ($A\subseteq \widebar{G}, B\subseteq G$). Now we introduce a few terminologies that will help us describe our DP algorithm later.

\begin{definition}[Partial solutions and $I$-optimality]
    \label{defn:partialsol}
A \emph{partial exchange solution} (or simply a \emph{partial solution}) refers to a pair of item subsets $(A',B')$ where $A'\subseteq \widebar{G}, B'\subseteq G$.  The \emph{weight} and \emph{profit} of the partial solution $(A',B')$ are defined as $W(A')-W(B')$ and $P(A')-P(B')$ respectively.

Let $I\subseteq [n]$ be an item subset. We say the partial solution $(A',B')$ is \emph{supported on} $I$ if $A'\cup B'\subseteq I$.
We say $(A',B')$ is \emph{$I$-optimal}, if there exists an optimal exchange solution $(A,B)$ such that $A' = A \cap I$ and $B' = B\cap I$.
\end{definition}

\begin{definition}[DP tables]
    \label{defn:dptable}
A \emph{DP table of size $L$} is an array  $q[-L\dd L]$ with entries $q[z]\in \Z\cup \{-\infty\}$ for $z\in \{-L,\dots,L\}$.\footnote{For brevity we call it size-$L$ despite its actual length being $(2L+1)$.} (By convention, assume $q[z]= -\infty$ for $|z|>L$.)
 We omit its index range and simply write $q[\, ]$ whenever its size is clear from context or is unimportant.

For an item subset $I\subseteq [n]$, we say $q[\,]$ is an \emph{$I$-valid DP table}, if for every entry $q[z]\neq -\infty$ there exists a corresponding partial solution $(A',B')$ supported on $I$ with weight $W(A')-W(B')=z$ and profit $P(A')-P(B')=q[z]$.
An $I$-valid DP table $q[\,]$ is said to be \emph{$I$-optimal} if it contains some entry $q[z]$ that corresponds to an $I$-optimal partial solution.
\end{definition}

For example, the trivial DP table with $q[0]=0, q[z]=-\infty (z\neq 0)$ is $\varnothing$-optimal (it contains the empty partial solution $(\varnothing,\varnothing)$).
In dynamic programming we gradually extend this $\varnothing$-optimal DP table to an $[n]$-optimal DP table which should contain an optimal exchange solution.
As a basic example, given an $I$-optimal DP table $q[-L\dd L]$, for $i\notin I$ we can obtain an $(I\cup \{i\})$-optimal DP table $q'[-L-w_i\dd L+w_i]$ in $O(L+w_i)$ time via the update rule $q'[z] \gets \max\{q[z],q[z\mp w_i]\pm p_i\}$ (where $\pm$ is $+$ if $i\in \widebar{G}$, or $-$ if $i\in G$).

Previous dynamic programming algorithms for $0$-$1$ Knapsack \cite{icalp21,chen2023faster,AxiotisT19,KellererP04} used the following standard lemma based on the SMAWK algorithm \cite{smawk}:
\begin{lemma}[Batch-updating items of the same weight]
    \label{lem:batch}
    Let $I\subseteq [n]$ and $J \subseteq \widebar{G}$ be disjoint item subsets, and all $j\in J$ have the same weight $w_j=w$. Suppose an upper bound $L'$ is known such that all $(I\cup J)$-optimal partial solutions $(A',B')$ satisfy $|W(A')-W(B')|\le L'$.

    Then, given an $I$-optimal DP table $q[-L\dd L]$, we can compute an $(I\cup J)$-optimal DP table $q'[-L'\dd L']$ in $O(L+L' +  |J|+ \ww\log \ww )$ time. 

    The same statement holds if the assumption $J\subseteq \widebar{G}$ is replaced by $J\subseteq G$.
\end{lemma}
\begin{proof}[Proof Sketch]
   Let $J = \{j_1,j_2,\dots,j_{|J|} \}\subseteq \widebar{G}$ be sorted so that $p_{j_1}>p_{j_2}>\dots>p_{j_{|J|}}$.
By the prefix property (\cref{lem:prefix}),  for function $Q(x):= p_{j_1}+p_{j_2}+\dots + p_{j_x}$ we know that $q'[z]:= \max_{x\ge 0}\{q[z-xw]+Q(x)\}$ gives an $(I\cup J)$-optimal DP table $q'[-L'\dd L']$.  Moreover, due to the size of the tables, it suffices to maximize over $x\ge 0$ where $xw \le L+L'$, i.e., $0\le x \le \lfloor (L+L')/w\rfloor$.

The task of computing $q'[\,]$ can be decomposed into $w$ independent subproblems based on the remainder $r = (z\bmod w)$. Each subproblem then translates to computing the $(\max,+)$-convolution of the subsequence of $q$ with indices $\equiv r\pmod{w}$ and the array $[Q(0),\dots,Q(\lfloor (L+L')/w\rfloor )]$, which can be done in linear time using SMAWK algorithm \cite{smawk} due to the concavity of $Q(\cdot)$ (i.e.,  $Q(x)-Q(x-1)\ge  Q(x+1)-Q(x)$ for all $x\ge 1$). (Readers unfamiliar with the result of \cite{smawk} may refer to \cref{app:smawk} for an overview.) Hence, the total time for running SMAWK is linear in the total size of these arrays, which is $O(L+L' + w)$. 

Finally we remark that we do not need to sort all items of $J$ in $O(|J|\log |J|)$ time at the beginning. By \cref{lem:prefix}, we only need to consider those with rank at most $2\ww$, which can be selected from $J$ and sorted in $O(|J|+\ww\log \ww)$ time \cite{BlumFPRT73}.
\end{proof}

\section{Algorithm for 0-1 Knapsack}
\label{sec:knapsack}
In this section we present our algorithm for 0-1 Knapsack (\cref{thm:knapsack-main}). In \cref{subsec:chenpartition}, we recall a crucial weight partitioning lemma from \cite{chen2023faster} based on fine-grained proximity, which naturally gives rise to a two-stage algorithm framework.
The second stage can be easily performed using previous techniques \cite{icalp21,chen2023faster} and is described in \cref{subsec:chenpartition}, while the first stage contains our main technical challenge and is described in \cref{subsec:firststage,sec:3.3,sec:propagation}: 
In \cref{subsec:firststage}, we give a rank partitioning lemma based on another proximity result.
Given this lemma, in \cref{sec:3.3} we abstract out a core subproblem called $\textsc{HintedKnapsackExtend}^+$, and describe how to implement the first stage of our algorithm assuming this core subproblem can be solved efficiently.
Our algorithm for $\textsc{HintedKnapsackExtend}^+$ will be described in \cref{sec:propagation}.

\subsection{Weight partitioning and the second-stage algorithm}
\label{subsec:chenpartition}

Chen, Lian, Mao, and Zhang \cite{chen2023faster}
recently used additive-combinatorial results of Bringmann and Wellnitz \cite{BringmannW21} to obtain several powerful structural lemmas involving the support size of two integer multisets $A,B$ avoiding non-zero common subset sums.
These structural results (called ``fine-grained proximity'' in \cite{chen2023faster}) allowed them to obtain faster knapsack algorithms than the earlier works \cite{icalp21,EisenbrandW20} based on $\ell_1$-proximity (\cref{eqn:l1proximity}) only.
Here we recall one of the key lemmas from \cite{chen2023faster}.\footnote{The original statement of \cite[Lemma 3.1]{chen2023faster} had a worse $\log N$ factor than the $\sqrt{\log N}$ factor in \cref{lem:proximitysupportsum}. By inspection of their proof, they actually proved the stronger version stated here in \cref{lem:proximitysupportsum}.}

\begin{lemma}[{\cite[Lemma 3.1]{chen2023faster}}, paraphrased]
    \label{lem:proximitysupportsum}
There is a constant $C$ such that the following holds. Suppose two multisets $A,B$ supported on $[N]$ satisfy \[|\supp(A)| \ge C \sqrt{N\log N} \] and \[\Sigma(B) \ge \frac{CN^2\sqrt{\log N}}{|\supp(A)|}.\]
Then, $\caS^*(A) \cap \caS^*(B) \neq \varnothing$.
\end{lemma}
Using this fine-grained proximity result, Chen, Lian, Mao, and Zhang obtained a weight partitioning lemma \cite[Lemma 4.1]{chen2023faster}, which is a key ingredient in their algorithm.
Our algorithm also crucially relies on this weight partitioning lemma in a similar way, but for our purpose we need to extend it from the two-partition version in \cite{chen2023faster} to $O(\log \ww)$-partition.\footnote{We remark that \cite[Lemma 5.3]{chen2023faster} also gave a three-partition extension of this lemma, but in a different way than what we need here.}

Recall the following notations from \cref{subsec:generalnotation}: $W(I)= \sum_{i\in I}w_i$, $\caW = \supp(\{w_1,w_2,\dots,w_n\}) \subseteq [\ww]$, and $I_{\caW'}:= \{ i\in [n]: w_i\in \caW'\}$.

 \begin{restatable}[Extension of {\cite[Lemma 4.1]{chen2023faster}}]{lemma}{weightpartition}
   \label{lem:weightpartition} 
   The set $\caW$ of input item weights can be partitioned in $O(n + \ww\log \ww)$ time into $\caW= \caW_1 \uplus \caW_2 \uplus \dots \uplus \caW_s$, where $s < \log_2(\sqrt{\ww})$, with the following property:
   
   Denote $\caW_{\le j} = \caW_1\cup  \dots\cup \caW_j$ and $\caW_{>j} = \caW \setminus \caW_{\le j}$. For every optimal exchange solution $(A,B)$ and every  $1\le j\le s$, 
    \begin{itemize}
        \item $|\caW_j| \le 4C\sqrt{\ww\log \ww } \cdot 2^{j}$,  and
            \item $W( A \cap  I_{\caW_{>j}}) \le 4C\ww^{3/2}/2^j$ and
             $W( B \cap  I_{\caW_{>j}}) \le 4C\ww^{3/2}/2^j$,
    \end{itemize}
    where $C$ is the universal constant from \cref{lem:proximitysupportsum}.
\end{restatable}
The proof of \cref{lem:weightpartition} is similar to that of the original two-partition version \cite[Lemma 4.1]{chen2023faster}, and is deferred to \cref{app:partition}.

Given this weight partitioning $\caW= \caW_1 \uplus \caW_2 \uplus \dots \uplus \caW_s$, our overall algorithm %
runs in two stages: in the first stage, we only consider items whose weights belong to $\caW_1$, and efficiently compute an $I_{\caW_1}$-optimal DP table (see \cref{defn:dptable}) by exploiting the small size of $\caW_1$. %
Then, the second stage of the algorithm updates the DP table using the remaining items $I_{\caW_2} \uplus \dots \uplus I_{\caW_{s}}$.
The second stage follows the same idea as \cite{chen2023faster} of using \cref{lem:weightpartition}  to trade off the size of the DP table and the number of linear-time scans needed to update the DP table. In contrast, the first stage is more technically challenging; we summarize it in the following lemma, and prove it in subsequent sections: 
\begin{lemma}[The first stage]
    \label{lem:firststage}
    Let $\caW_1\subseteq [\ww]$ from \cref{lem:weightpartition} be given.
    Then we can compute an $I_{\caW_1}$-optimal DP table in $O(n+\ww^2\log^4\ww)$  time.
\end{lemma}
The overall $O(n+\ww^2\log^4 \ww)$ algorithm for $0$-$1$ Knapsack then follows from \cref{lem:firststage} and \cref{lem:weightpartition}, using arguments similar to \cite{chen2023faster}.

\begin{proof}[Proof of \cref{thm:knapsack-main} assuming \cref{lem:firststage}]
    Use \cref{lem:weightpartition} to obtain the weight partition 
$\caW= \caW_1 \uplus \caW_2 \uplus \dots \uplus \caW_s$, where $s < \log_2(\sqrt{\ww})$. Define a size upper bound $L_j:=4C\ww^{3/2}/2^{j}+\ww$ for $1\le j\le s$.
Note that for every optimal exchange solution $(A,B)$ and every $1\le j \le s$ we have 
\begin{align}
   \lvert W(A\cap I_{\caW_{\le j}}) - W(B\cap I_{\caW_{\le j}}) \rvert &= \big \lvert \big(W(A)-W(B)\big ) - \big (W(A\cap I_{\caW_{>j}}) - W(B\cap I_{\caW_{>j}}) \big )\big \rvert \nonumber \\
   & \le |W(A)-W(B)| + \lvert W(A\cap I_{\caW_{>j}}) - W(B\cap I_{\caW_{>j}})\rvert \nonumber\\
   & < \ww + 4C\ww^{3/2}/2^{j} \nonumber \\
   & = L_j, \label{eqn:partialsolutionbound}
\end{align}
where the last inequality follows from \cref{eqn:sumclose} and \cref{lem:weightpartition}.

We run the first-stage algorithm of \cref{lem:firststage}, and obtain an $I_{\caW_1}$-optimal DP table. By \cref{eqn:partialsolutionbound}, we can shrink the size of this DP table to $L_1$ without losing its $I_{\caW_1}$-optimality. Then, we repeatedly apply the following claim:
    \begin{claim}
        \label{claim:eachphase}
        For $2\le j\le s$, given an $I_{\caW_{\le j-1}}$-optimal DP table of size $L_{j-1}$,  we can compute an $I_{\caW_{\le j}}$-optimal DP table of size $L_j$ in $O(|\caW_j|(L_{j-1}+\ww\log \ww) + |I_{\caW_j}|)$ time.
    \end{claim}
    \begin{proof}[Proof of \cref{claim:eachphase}]
        We need to update the DP table using items from $I_{\caW_j}$.
         Partition them into $I_{\caW_j}^+ = I_{\caW_j} \cap \widebar{G}$ and $I_{\caW_j}^- = I_{\caW_j} \cap G$ (recall from \cref{subsec:maximalprefix} that $G$ denotes the greedy solution). We first update the DP table with the ``positive items'' $I_{\caW_j}^+$.
       By  a similar proof to \cref{eqn:partialsolutionbound}, we know 
       \[\lvert W(B\cap I_{\caW_{\le j-1}}) - W(A\cap I_{\caW_{\le j}}) \rvert \le L_{j-1}\] holds for any optimal exchange solution $(A,B)$.
Hence,  we can iterate over $w\in \caW_j$, and use the batch-update lemma based on SMAWK (\cref{lem:batch}) to update the size-$L_{j-1}$ DP table with the weight-$w$ items $I_{\{w\}} \cap \widebar{G}$, in $O(L_{j-1} + |I_{\{w\}}\cap \widebar{G}| + \ww\log \ww)$ time.
In the end we obtain a $(I_{\caW_{\le j-1}} \cup I_{\caW_j}^+)$-optimal DP table of size $L_{j-1}$, in total time $\sum_{w\in \caW_j}O(L_{j-1} + |I_{\{w\}}\cap \widebar{G}| + \ww\log \ww) = O(|\caW_j|(L_{j-1}+\ww\log \ww) + |I_{\caW_j}\cap \widebar{G}|)$.

Then, we update the $(I_{\caW_{\le j-1}} \cup I_{\caW_j}^+)$-optimal DP table with the ``negative items'' $I_{\caW_j}^-$, and obtain an $I_{\caW_{\le j}}$-optimal DP table. This algorithm is symmetric to the positive case described above, and similarly runs in time $O(|\caW_j|(L_{j-1}+\ww\log \ww) + |I_{\caW_j}\cap G|)$.

In the end, we shrink the size of the obtained  DP table to $L_{j}$. By \cref{eqn:partialsolutionbound}, this does not affect the $I_{\caW_{\le j}}$-optimality of the DP table. The total running time is 
$O(|\caW_j|(L_{j-1}+\ww\log \ww) + |I_{\caW_j}|)$.
    \end{proof}
In the main algorithm, starting from the $I_{\caW_1}$-optimal DP table, we sequentially apply \cref{claim:eachphase} for $j=2,3,\dots,s$, and in the end we can obtain an $I_{\caW_{\le s}}$-optimal DP table. Note that $I_{\caW_{\le s}} = [n]$, so the final DP table contains an optimal exchange solution. The total time for applying \cref{claim:eachphase} is (up to a constant factor)
\begin{align*}
 &   \sum_{j=2}^s \big (|\caW_j|(L_{j-1}+\ww\log \ww) + |I_{\caW_j}| \big ) \\
 \le  \ & \sum_{j=2}^s 4C\sqrt{\ww\log \ww}\cdot 2^j\cdot (4C\ww^{3/2}/2^{j-1} +2\ww\log \ww) + \sum_{j=2}^s|I_{\caW_j}|\\
    =  \ &  \sum_{j=2}^s \big (  32C^2 \ww^2\sqrt{\log \ww} + 8C\cdot 2^j\cdot (\ww\log \ww)^{3/2}\big ) + \sum_{j=2}^s|I_{\caW_j}|\\
 \le \ &  O(\ww^2 (\log \ww)^{3/2} + n). \tag{by $s< \log_2(\sqrt{\ww})$}
\end{align*}
The time complexity of the entire algorithm is dominated by the $O(n + \ww^2 \log^4 \ww)$ running time of the first stage (\cref{lem:firststage}).
\end{proof}

\subsection{Rank partitioning}
\label{subsec:firststage}
Given $\caW_1\subseteq [\ww]$ of size $|\caW_1|\le O(\sqrt{\ww\log \ww})$ from \cref{lem:weightpartition}, we  partition the items whose weights belong to $\caW_1$ into dyadic groups based on their ranks (\cref{defn:rank}), as follows:

\begin{definition}[Rank partitioning]
    \label{defn:rankpartitioning}
   Let  $k = \lceil \log_2( 2\ww+1)\rceil$. For each $1\le j\le k$, define item subsets
\[ J^+_j := \{ i \in \widebar{G} \cap I_{\caW_1} : 2^{j-1}\le \rank(i)\le   2^j-1\}, \]
and
\[ J^-_j := \{ i \in G \cap I_{\caW_1} : 2^{j-1}\le \rank(i)\le  2^j-1\}. \]
Note that $J_1^+ \uplus J_1^- \uplus \dots \uplus  J_k^+\uplus J_k^-$ form a partition of $\{i\in I_{\caW_1}: \rank(i)\le 2^k-1\}$.

Denote
\[ J^+_{\le j} = J^+_{1} \cup \dots \cup J^{+}_{j},\]
 and
\[ J^-_{\le j} = J^-_{1} \cup \dots \cup J^{-}_{j}.\]
\end{definition}

Note the the rank partitioning defined in \cref{defn:rankpartitioning} can be computed in $O(n + \ww \log \ww)$ time.

Our rank partitioning is motivated by the following additive-combinatorial \cref{lem:exchange-l0proximity}, 
which can be derived from the results of Bringmann and Wellnitz \cite{BringmannW21}.
Recall the $r$-\emph{support} $\supp_r(X)$ of a multiset $X$ is the set of items in $X$ with multiplicity at least $r$.

\begin{restatable}{lemma}{exchangegeneral}
    \label{lem:exchange-l0proximity}
There is a constant $C$ such that the following holds. Suppose two multisets $A,B$ supported on $[N]$ satisfy 
\begin{equation} |\supp_r(A)| \ge C \sqrt{N/r}\cdot \sqrt{\log (2N)} \label{eqn:exchangelemmacond1} \end{equation}
for some $r\ge 1$,
and 
\begin{equation}
    \label{eqn:exchangelemmacond2}
\Sigma(B) \ge \Sigma(A)-N.
\end{equation}
Then, $\caS^*(A) \cap \caS^*(B) \neq \varnothing$.
\end{restatable}
\cref{lem:exchange-l0proximity} is partly
inspired by \cite[Lemma 3.2]{chen2023faster} which generalized their fine-grained proximity result (\cref{lem:proximitysupportsum}) from $\supp(A)$ to $\supp_r(A)$.\footnote{Their generalization of \cref{lem:proximitysupportsum} is not applicable in our first-stage algorithm. Note that \cref{lem:exchange-l0proximity} is incomparable to \cref{lem:proximitysupportsum} even when $r=1$.}
 We include a proof of \cref{lem:exchange-l0proximity} in \cref{appendix:comb}.

 Using \cref{lem:exchange-l0proximity}, we obtain the following structural lemma for the rank partitioning. Recall the definition of $I_{\caW_1}$-optimal partial solutions from \cref{defn:partialsol}.

\begin{lemma}[Rank partitioning structural lemma]
    \label{lem:rankstructural}

       For a universal constant $C$, the partition $J_1^+ \uplus J_1^- \uplus \dots \uplus  J_k^+\uplus J_k^-\subseteq I_{\caW_1}$ from \cref{defn:rankpartitioning} satisfies the following properties for every $I_{\caW_1}$-optimal partial solution $(A',B')$:
\begin{enumerate}
        \item $A' \subseteq J^+_{\le k}$ and $B' \subseteq J^-_{\le k}$.
            \label{item:1}
    \item For all $1\le j\le k$,
    $|A' \cap J^+_{\le j}| \le m_j$ and  $|B' \cap J^-_{\le j}| \le m_j$, where
            \label{item:2}
        \begin{equation}
            \label{eqn:defmj}
         m_j := C\cdot 2^{j/2} \cdot \sqrt{\ww\log (2\ww)}.
        \end{equation}
        \item For all $1\le j\le k$,
       \[ |\{ w\in \caW_1 : I_{\{w\}} \cap J_{j-1}^+ \subseteq A' \text{ and }I_{\{w\}} \cap J_{j}^+\neq \varnothing\}| \le b_j,\]
       and similarly
       \[ |\{ w\in \caW_1 : I_{\{w\}} \cap J_{j-1}^- \subseteq B' \text{ and }I_{\{w\}} \cap J_{j}^-\neq \varnothing\}| \le b_j,\]
       where
       \begin{equation}
           \label{eqn:defbj}
b_j:=         C\cdot 2^{-j/2}\cdot \sqrt{\ww\log (2\ww)},
       \end{equation} 
and $J_0^+:=J_0^-:=\varnothing$.
            \label{item:3}
\end{enumerate}
\end{lemma}
\begin{proof}
    Let $(A,B)$ be an optimal exchange solution such that $A\cap I_{\caW_1}=A'$ and $B\cap I_{\caW_1}=B'$.  In the following we will only prove the claimed properties about $A'$, and the properties about $B'$ can be proved similarly.

    For \cref{item:1}, note that $J_{\le k}^+ = \{i\in \widebar{G}\cap I_{\caW_1} :  \rank(i)\le 2^k-1\}$, and we have $2^k-1 \ge 2\ww$ by the definition of $k$. 
    By the prefix property (\cref{lem:prefix}), we have $\rank(i)\le 2\ww$ for all $i\in A'$, and thus $A' \subseteq J_{\le k}^+$.

For \cref{item:2}, we apply \cref{lem:exchange-l0proximity} with $N:=\ww$ to the multisets $\wts(A)$ and $\wts(B)$, which should not share any common non-zero subset sum (\cref{eqn:nocommon}). Since $|W(A)-W(B)|< \ww$ (\cref{eqn:sumclose}), \cref{lem:exchange-l0proximity} implies for all $r\ge 1$ that
\begin{equation}
|\supp_r(\wts(A))| < C_0 \sqrt{\ww/r}\cdot \sqrt{\log (2\ww)}
\label{eqn:tempsuppr}
\end{equation}
for some universal constant $C_0$.
Hence,
\begin{align*}
    |A'\cap J_{\le j}^+| &= |\{ i\in A': 1\le \rank(i)\le 2^j-1\}|\\
    & = \sum_{w\in \caW_1}\min\{ 2^{j}-1, |A'\cap I_{\{w\}}|\}\\ %
    & = \sum_{r=1}^{2^j-1}|\supp_r(\wts(A' ))|\\
    & \le  \sum_{r=1}^{2^j-1}|\supp_r(\wts(A))|\\
    & \le \sum_{r=1}^{2^j-1}C_0 \sqrt{\ww/r}\cdot \sqrt{\log (2\ww)} \tag{by \cref{eqn:tempsuppr}}\\
    & < C_0 \sqrt{\ww}\cdot 2\sqrt{2^j-1}\cdot \sqrt{\log (2\ww)}\\
    & < C2^{j/2}\sqrt{\ww\log(2\ww)} 
\end{align*}
for some universal constant $C$.

For \cref{item:3}, 
the set under consideration is a subset of $\caW_1$, and in the $j=1$ case we can simply bound its size using \cref{lem:weightpartition} as $\le |\caW_1| \le 8C_1 \sqrt{\ww\log \ww} \le b_1$,
provided the constant $C$ in the definition of $b_1$ (\cref{eqn:defbj}) is large enough. Now it remains to consider $2\le j\le k$, and we have to bound the number of 
$w\in \caW_1$ such that $I_{\{w\}} \cap J_{ j-1}^+ \subseteq A' $ and $I_{\{w\}} \cap J_{j}^+\neq \varnothing$. 
For any such $w$, by definition of our rank partitioning (\cref{defn:rankpartitioning}),  $I_{\{w\}} \cap J_{j}^+\neq \varnothing$ means that $I_{\{w\}} \cap \widebar{G}$ contains some item $i$ with $\rank(i)\ge 2^{j-1}$, and 
hence $|I_{\{w\}} \cap J_{j-1}^+| = 2^{j-2}$.
Then from $I_{\{w\}} \cap J_{j-1}^+ \subseteq A'$ we get $|I_{\{w\}} \cap A'|\ge 2^{j-2}$,  or equivalently, $w \in \supp_{2^{j-2}}(\wts(A'))$. Thus, the number of such $w$ is at most
\begin{align*}
 |\supp_{2^{j-2}}(\wts(A'))| &\le |\supp_{2^{j-2}}(\wts(A))| \\
& < C_0 \sqrt{\ww/(2^{j-2})}\cdot \sqrt{\log(2\ww)}  \tag{by \cref{eqn:tempsuppr}}\\
& \le  C \sqrt{\ww} \cdot 2^{-j/2} \sqrt{\log(2\ww)} \\
&= b_j
\end{align*}
for some large enough constant $C$.
\end{proof}

\subsection{The first-stage algorithm via hinted dynamic programming}
\label{sec:3.3}

Based on our rank partitioning
$J_1^+ \uplus J_1^- \uplus \dots \uplus  J_k^+\uplus J_k^-\subseteq I_{\caW_1}$, $k = \lceil \log_2( 2\ww+1)\rceil$ (\cref{defn:rankpartitioning}) and its
structural lemma (\cref{lem:rankstructural}), our first-stage algorithm uses dynamic programming and runs in $k$ phases.
At the beginning of the $j$-th phase ($1\le j\le k$), we have a $(J_{\le j-1}^+ \cup J_{\le j-1}^-)$-optimal DP table, and we first update it with the ``positive items'' $J_{j}^+$ to obtain a $(J_{\le j}^+ \cup J_{\le j-1}^-)$-optimal DP table, and then update it with the ``negative items'' $J_{j}^-$ to obtain a $(J_{\le j}^+ \cup J_{\le j}^-)$-optimal DP table.
We will adjust the size of the DP table throughout the $k$ phases based on \cref{item:2} of \cref{lem:rankstructural}.
This is similar to the second-stage algorithm from \cref{subsec:chenpartition}, except that in \cref{subsec:chenpartition} the DP table is shrinking whereas here it will be expanding.

To implement the DP efficiently, we crucially rely on \cref{item:3} of \cref{lem:rankstructural}, which gives an upper bound on the ``active support'' of the weights of items in every partial solution in the current DP table.
More specifically, consider an $I_{\caW_1}$-optimal partial solution $(A',B')$ and its restriction $(A'',B'')$ where $A''= A' \cap J_{\le j-1}^+, B''=B' \cap J_{\le j-1}^-$.
Then \cref{item:3} of \cref{lem:rankstructural} implies that the items in $A'\setminus  A''$ (or $B'\setminus  B''$) can only have at most $b_j$ distinct weights.
This means that, for any partial solution $(A'',B'')$ in the DP table at the end of phase $j-1$, in order to extend it to an $I_{\caW_1}$-optimal partial solution $(A',B')$ in future phases, we only need to update it with items from these $b_j$ weight classes determined by \cref{item:3} of \cref{lem:rankstructural}. 
This idea is called \emph{witness propagation}, and was originally introduced by Deng, Mao, and Zhong \cite{dmz23} in the context of unbounded knapsack-type problems. Implementing this idea in the more difficult $0$-$1$ setting is a main technical contribution of this paper.

In the rest of this section, we will introduce a few more definitions to help use formally describe our algorithm, and we will abstract out a core subproblem called 
$\textsc{HintedKnapsackExtend}^+$ which captures the aforementioned idea of witness propagation. Then we will show how to implement our first-stage algorithm and prove \cref{lem:firststage}, assuming $\textsc{HintedKnapsackExtend}^+$ can be solved efficiently.

In the following definition, we augment each entry of the DP table with hints, which contain the weight classes from which we need to add items when we update this entry, as we just discussed.

\begin{definition}[Hinted DP tables]
    \label{defn:hinteddptable}
    A \emph{hinted DP table} is a DP table $q[\, ]$ where each entry $q[z]\neq -\infty$ is annotated with two sets $S^+[z],S^-[z]\subseteq \caW_1$. We say the table has \emph{positive hint size} $b$ if $|S^+[z]|\le b$ for all $z$, and has \emph{negative hint size} $b$ if $|S^-[z]|\le b$ for all $z$.

 For an item subset $J\subseteq I_{\caW_1}$, we say a hinted DP table $q[\, ]$ is \emph{hinted-$J$-optimal}, if $q[\,]$ is $J$-valid (see \cref{defn:dptable}), and it has an entry $q[z]$ such that both of the following hold:
 \begin{enumerate}
    \item \label{item:hinteditem1} There exists an $I_{\caW_1}$-optimal partial solution $(A',B')$ such that $W(A'\cap J)-W(B'\cap J)=z$ and $P(A'\cap J)-P(B'\cap J)=q[z]$. 
    \item \label{item:hinteditem2} Every $I_{\caW_1}$-optimal partial solution $(A',B')$ with $W(A'\cap J)-W(B'\cap J)=z$ should satisfy $A'\setminus J \subseteq I_{S^+[z]}$ and $B'\setminus J\subseteq I_{S^-[z]}$.
 \end{enumerate}
\end{definition}
Note that if a hinted DP table is hinted-$J$-optimal, then in particular it is $J$-optimal in the sense of \cref{defn:dptable} (due to \cref{item:hinteditem1} of \cref{defn:hinteddptable}).

The following lemma summarizes each of the $k= \lceil \log_2( 2\ww+1)\rceil$ phases in our first-stage algorithm.

\begin{lemma}
    \label{lem:propagate-one-phase}
    Let $k, m_j, b_j$ be defined as in \cref{lem:rankstructural}.
    Let $L_j := m_j\cdot \ww$.
    For every $1\le j\le k$, the following hold:
\begin{enumerate}
    \item 
    Given a hinted-$(J_{\le j-1}^+ \cup J_{\le j-1}^-)$-optimal DP table of size $L_{j-1}$ with positive  and negative hint size $b_{j}$, we can compute a hinted-$(J_{\le j}^+ \cup J_{\le j-1}^-)$-optimal DP table of size $L_j$ with positive hint size $b_{j+1}$ and negative hint size $b_j$, in $O(L_{j}b_j\cdot \log^2(L_jb_j) + |J^+_{j}|)$ time.
    \label{item:phase01}
    \item 
    Given a hinted-$(J_{\le j}^+ \cup J_{\le j-1}^-)$-optimal DP table of size $L_j$
    with positive hint size $b_{j+1}$ and negative hint size $b_j$,
 we can compute a hinted-$(J_{\le j}^+ \cup J_{\le j}^-)$-optimal DP table of size $L_j$ with positive and negative hint size $b_{j+1}$, in $O(L_{j}b_j\cdot \log^2(L_jb_j)+ |J^-_{j}|)$ time.
    \label{item:phase02}
\end{enumerate}
\end{lemma}

\cref{lem:propagate-one-phase} immediately implies our overall first-stage algorithm:
\begin{proof}[Proof of \cref{lem:firststage} assuming \cref{lem:propagate-one-phase}]
    We start with the trivial hinted DP table with $q[0]=0$ and $S^+[0] = S^-[0] = \caW_1$ (padded with $q[z]=-\infty$ for $z\in \{-L_0,\dots,L_0\}\setminus \{0\}$).  By definition, $q[\, ]$ is hinted-$\varnothing$-optimal, and has positive and negative hint size $|\caW_1| \le b_1$.
Then, we iteratively perform phases $j=1,2,\dots,k$, where in phase $j$ we first apply \cref{item:phase01} of \cref{lem:propagate-one-phase}, and then apply \cref{item:phase02} of \cref{lem:propagate-one-phase}.
In the end of phase $k$, we have obtained a  hinted-$(J_{\le k}^+ \cup J_{\le k}^-)$-optimal DP table. By \cref{item:1} of \cref{lem:rankstructural}, it is an $I_{\caW_1}$-optimal DP table as desired.

The total time of applying \cref{lem:propagate-one-phase} is (up to a constant factor)
\begin{align*}
   & \sum_{j=1}^k \big (L_{j}b_j\cdot \log^2(L_jb_j) + |J^+_{j}| + |J^-_{j}|\big )\\
     \le  \ & \sum_{j=1}^k  m_j \ww b_j \log^2(m_j \ww b_j)\, + n\\
     = \ & \sum_{j=1}^k C^2 \ww^2\log (2\ww)   \log^2\big (C^2 \ww^2\log (2\ww)   \big ) \, + n \tag{by \cref{eqn:defmj,eqn:defbj}}\\
     \le \ & O(\ww^2 \log^4 \ww + n). \tag{by $k= \lceil \log_2( 2\ww+1)\rceil$}
\end{align*}
\end{proof}

It remains to prove \cref{lem:propagate-one-phase}. In the following, we will reduce it to a core subproblem called $\textsc{HintedKnapsackExtend}^+$, which captures the task of updating a hinted size-$L$ DP table with positive hint size $b$ using ``positive items'' whose weights come from some positive integer set $U$ (here we can think of $U = \caW_1$).
Similarly to the proof of the batch-update lemma (\cref{lem:batch}) based on SMAWK, here we also use a function
$Q_w\colon \Z_{\ge 0} \to\Z$ to represent the total profit of taking the top-$x$ items of weight $w$.

\begin{prob}[$\textsc{HintedKnapsackExtend}^+$]
    \label{prob:prob3}
    
    Let $U \subseteq \caW_1$. For every $w\in U$,  suppose $Q_w\colon \Z_{\ge 0} \to\Z$ is a 
     concave function with $Q_w(0)=0$ that can be evaluated in constant time.
    We are given a DP table $q[-L\dd L]$ (where $q[i]\in \Z\cup \{-\infty\}$), annotated with $S[-L\dd L]$ where $S[i]\subseteq U$.

    Consider the following optimization problem for each $-L\le i\le L$: find a solution vector $\vec x[i] \in \Z_{\ge 0}^{U}$ that maximizes the total profit
\begin{equation}
    \label{eqn:tomaximize}
 r[i]:= q\big [z[i]\big ]  + \sum_{w\in U}Q_w(x[i]_w),
\end{equation}
where $z[i] \in \Z$ is uniquely determined by
\begin{equation}
    \label{eqn:subjectto}
z[i] + \sum_{w\in U}w\cdot x[i]_w = i.
\end{equation}

The task is to solve this optimization problem for each $-L\le i\le L$ \emph{with the following relaxation}:
\begin{itemize}
    \item 
If \emph{all} maximizers $(\vec x[i],z[i])$ of \cref{eqn:tomaximize} (subject to \cref{eqn:subjectto}) satisfy 
\begin{equation}
    \label{eqn:supportcontain}
\supp(\vec x[i]) \subseteq S\big [z[i]\big ],
\end{equation}
then we are required to correctly output a maximizer for $i$.
\item Otherwise, we are allowed to output a suboptimal solution for $i$.
\end{itemize}
\end{prob}

\begin{remark}
We give a few remarks to help get a better understanding of \cref{prob:prob3}:    \begin{enumerate}
\item 
In \cref{eqn:subjectto}, $z[i] \le i$ must hold, since $w\in U \subseteq [\ww]$ is always positive and $\vec{x}[i]$ is a non-negative vector.
    \item 
If we do not have the relaxation based on hints $S[i]$, then \cref{prob:prob3} becomes a standard problem solvable in $O(|U| L)$ time using SMAWK algorithm (basically, repeat the proof of \cref{lem:batch} for every $w\in U$; see also \cite{icalp21,chen2023faster,AxiotisT19,KellererP04}).
        \item 
Under this relaxation, without loss of generality, we can assume the output of \cref{prob:prob3} always satisfies $\supp(\vec x[i]) \subseteq S\big [z[i]\big ]$ (\cref{eqn:supportcontain}) for all $-L\le i\le L$. (If we had to output an $\vec x[i]$ that violates \cref{eqn:supportcontain}, then we must be in the ``otherwise'' case for $i$, and should be allowed to output anything). In particular, if $|S[i]|\le b$ for all $-L\le i\le L$, then we can assume the output $\vec x[-L\dd L]$ of \cref{prob:prob3} has description size $O(bL)$ words.
\item 
Note that \cref{prob:prob3} is different from (and easier than) 
the task of maximizing \cref{eqn:tomaximize} for every $i$ subject to \cref{eqn:supportcontain}. The latter version would make a cleaner definition, but it is a  harder problem which we do not know how to solve. 
    \end{enumerate}
\end{remark}

The following \cref{lem:prob3largeb} summarizes our algorithm for \cref{prob:prob3}, which will be given in \cref{sec:propagation}.
\begin{restatable}{theorem}{algolargeb}
    \label{lem:prob3largeb}
   $\textsc{HintedKnapsackExtend}^+$ (\cref{prob:prob3}) with
   $|S[i]|\le b$ for all $-L\le i \le L$ can be solved deterministically in $O(Lb\log^2(Lb))$ time.
\end{restatable}
Finally, we show how to prove \cref{lem:propagate-one-phase} using \cref{lem:prob3largeb}.
\begin{proof}[Proof of \cref{lem:propagate-one-phase} assuming \cref{lem:prob3largeb}]
    Here we only prove \cref{item:phase01} of \cref{lem:propagate-one-phase} (\cref{item:phase02} can be proved similarly in the reverse direction).
    Given a hinted-$(J_{\le j-1}^+ \cup J_{\le j-1}^-)$-optimal DP table $q[-L_{j-1}\dd L_{j-1}]$ with positive hint size $|S^+[i]|\le b_j$ and negative hint size $|S^-[i]|\le b_j$, our task is to  compute a hinted-$(J_{\le j}^+ \cup J_{\le j-1}^-)$-optimal DP table $r[\,]$,  by updating $q[\,]$ with items from $J_{j}^+$.

   As usual, partition $J_j^+$ into weight classes: for each $w\in \caW_1$, let the weight-$w$ items $J_{j}^+ \cap I_{\{w\}} = \{i_1,i_2,\dots,i_{m}\}$ be sorted so that $p_{i_1}>p_{i_2}>\dots>p_{i_m}$,
    and define their top-$x$ total profit $Q_w(x):= p_{i_1}+p_{i_2}+\dots + p_{i_x}$, which is a concave function in $x$. %
After preprocessing in $O(|J_{j}^+| + |\caW_1|\ww\log \ww)$ total time, $Q_w(\cdot )$ can be evaluated in constant time for all $w\in \caW_1$.%

Then, we enlarge the DP table $q[-L_{j-1}\dd L_{j-1}]$ to $q[-L_{j}\dd L_j]$ by padding dummy entries $-\infty$, and define a $\textsc{HintedKnapsackExtend}^+$ instance on functions $Q_w$ and DP table $q[-L_{j}\dd L_{j}]$ with hints $S[-L_{j}\dd L_{j}] := S^+[-L_{j}\dd L_{j}]$ (recall $S^+[i]\subseteq \caW_1$ and $|S^+[i]|\le b_j$). We run the algorithm from \cref{lem:prob3largeb} to solve this instance  in $O(L_j b_j\log^2(L_jb_j))$ time, and obtain solution vectors $\vec x[i]\in \Z_{\ge 0}^{\caW_1}$ for all $-L_j\le i \le L_j$. Recall from \cref{eqn:tomaximize,eqn:subjectto} that solution vector $\vec x[i]$ achieves total profit
\[
 r[i]:= q\big [z[i]\big ]  + \sum_{w\in \caW_1}Q_w(x[i]_w),
\]
where 
$z[i] + \sum_{w\in \caW_1}w\cdot x[i]_w = i$.

In the following, we will first show that $r[-L_j\dd L_j]$ satisfies \cref{item:hinteditem1} of \cref{defn:hinteddptable} (and thus is an $(J_{\le j}^+ \cup J_{\le j-1}^-)$-optimal DP table). Then we will compute new hints $T^+[-L_j\dd L_j],T^-[-L_j\dd L_j]$ to satisfy \cref{item:hinteditem2} of \cref{defn:hinteddptable}, making $r[-L_j\dd L_j]$ a hinted-$(J_{\le j}^+ \cup J_{\le j-1}^-)$-optimal DP table.

\paragraph*{Optimality.}
Since  
$q[\, ]$ 
is a hinted-$(J_{\le j-1}^+ \cup J_{\le j-1}^-)$-optimal DP table by assumption, by \cref{defn:hinteddptable} (\cref{item:hinteditem1}) there exists an $I_{\caW_1}$-optimal partial solution $(A',B')$ such that $q[\, ]$ contains partial solution $(A'',B'')$ where $A'':=A' \cap J_{\le j-1}^+, B'' := B' \cap J_{\le j-1}^-$, that is, we have $q[z''] = P(A'')-P(B'')$ for $z'':=W(A'')-W(B'')$.
Now consider the partial solution $(\widehat A'',B'')$ where $\widehat A'':= A' \cap J_{\le j}^+$, and denote their difference by $Y:= \widehat A'' \setminus A'' =  A' \cap J_{j}^+$. 
Encode the items in $Y$ by a vector $\vec y\in \Z_{\ge 0}^{\caW_1}$ where $y_w=|I_{\{w\}} \cap Y|$, and by the prefix property (\cref{lem:prefix}) we must have $P(Y) = \sum_{w\in \caW_1} Q_w(y_w)$, so $q[z''] + \sum_{w\in \caW_1}Q_w(y_w) = P(\widehat A'')-P(B'')$.
Now the goal is to show that $(\widehat A'',B'')$ indeed survives in the solution of the $\textsc{HintedKnapsackExtend}^+$ instance, i.e., for $\hat{\imath}:=W(\widehat A'')-W(B'')$ we want to show $r[\hat{\imath}] = q[z''] + \sum_{w\in \caW_1}Q_w(y_w)$.

First we verify that $\hat{\imath}$ is contained in the index range $[-L_j\dd L_j]$ of the returned table: by \cref{lem:rankstructural} (\cref{item:2}) we have $W(\widehat A'')\le | A'\cap J_{\le j}^+|\cdot  \ww \le m_j\ww = L_j$, and similarly $W(B'') \le L_{j-1}$, so $|\hat{\imath}| = |W(\widehat A'')-W(B'')| \le L_j$ as desired. 
Now, if $r[ \hat{\imath}] \neq q[z''] + \sum_{w\in \caW_1}Q_w(y_w)$, then by the definition of \cref{prob:prob3} there are only two possibilities:
\begin{itemize}
    \item Case 1: $(\vec y,z'')$ is not a maximizer for $r[\hat{\imath}]$.
    
   This means there is a solution $(\vec x^\star, z^\star)$ with strictly better total profit $q[z^\star]+\sum_{w\in \caW_1}Q_w(x^\star_w) > q[z''] + \sum_{w\in \caW_1}Q_w(y_w)$. 
    Let $X^\star\subseteq J_j^+$ be the item set encoded by $\vec x^\star$, and let $q[z^\star]$ correspond to the partial solution $(A^\star,B^\star)$ supported on $(J_{\le j-1}^+\cup J_{\le j-1}^-)$. Then the partial solution $(A^\star \cup X^\star, B^\star)$ has the same weight as
    $(\widehat A'' , B'')$ but achieves strictly higher profit (both of them are supported on $J_{\le j}^+\cup J_{\le j-1}^-$).
      This contradicts the $I_{\caW_1}$-optimality of $(A',B')$ by an exchange argument. %

        \item Case 2:  $(\vec y, z'')$ is a maximizer for $r[\hat{\imath}]$, but there is also another maximizer $(\vec x^\star, z^\star)$ for $r[\hat{\imath}]$ that violates the support containment condition $\supp(\vec x^\star)\subseteq S[z^\star] = S^+[z^\star]$ (\cref{eqn:supportcontain}).
            
            Let $X^\star \subseteq J_j^+$ and $(A^\star,B^\star)$ be defined in the same way as in Case 1.
        Similarly to the discussion in Case 1, here we know that there is an alternative $I_{\caW_1}$-optimal partial solution $(\tilde A', \tilde B') := \big ((A' \setminus \widehat A'')\cup  (A^\star \cup X^*),\, (B'\setminus B'')\cup B^\star\big )$ that achieves the same weight and profit as $(A',B')$.

         By the violation of \cref{eqn:supportcontain} we have $\supp(\wts(X^\star)) = \supp(\vec x^\star) \not \subseteq S^+[z^\star]$, that is, $X^\star \not \subseteq I_{S^+[z^\star]}$.
        Since $\tilde A' \setminus (J_{\le j-1}^+\cup J_{\le j-1}^-) \supseteq X^\star $, this means $\tilde A' \setminus (J_{\le j-1}^+\cup J_{\le j-1}^-) \not \subseteq I_{S^+[z^\star]}$.
         This violates \cref{item:hinteditem2} of \cref{defn:hinteddptable} for entry $q[z^\star]$ and the $I_{\caW_1}$-optimal partial solution $(\tilde A',\tilde B')$, and hence contradicts the assumption that $q[\,]$ (with hints $S^+[\,], S^-[\,]$) is a hinted-$(J_{\le j-1}^+ \cup J_{\le j-1}^-)$-optimal DP table. \end{itemize}
This finishes the proof that $r[\hat{\imath}]= q[z''] + \sum_{w\in \caW_1}Q_w(y_w)$, meaning that $r[-L_j\dd L_j]$ is indeed an $(J_{\le j}^+ \cup J_{\le j-1}^-)$-optimal DP table.

For the rest of the proof, without loss of generality we assume $(\vec x[\hat{\imath}], z[\hat{\imath}]) = (\vec y,z'')$ (since we could have started the proof with the $I_{\caW_1}$-optimal partial solution $(A',B')$ which would produce a $\vec{y}$ that coincides with $\vec x[\hat{\imath}]$).

\paragraph*{New hints for hinted-optimality.}
Now we describe how to compute new hint arrays $T^+[-L_j\dd L_j]$, $T^-[-L_j\dd L_j]$ to annotate to the DP table $r[-L_j\dd L_j]$.

 We will crucially use \cref{item:3} of \cref{lem:rankstructural}, which states that for every $I_{\caW_1}$-optimal partial solution $( A',  B')$,
\begin{equation}
    \label{eqn:temphintsize}
        |\{ w\in \caW_1 : I_{\{w\}} \cap J_{j}^+ \subseteq  A' \text{ and }I_{\{w\}} \cap J_{j+1}^+\neq \varnothing\}| \le b_{j+1}.
\end{equation}

Given the solutions $(\vec x[i], z[i])$ returned by \cref{lem:prob3largeb} (where each $\vec x[i]$ encodes an item subset of $J_{ j}^+$),
 we do the following for every $-L_j\le i\le L_j$:
\begin{itemize}
        \item  Define hints
\begin{align*}
 T^-[i] &:= S^-\big [z[i]\big ], \\
 T^+[i] & := \{w \in \caW_1: |I_{\{w\}} \cap J_{j}^+| = x[i]_w   \text{ and }I_{\{w\}}\cap J_{j+1}^+ \neq \varnothing\}.
\end{align*}
    \item If $|T^+[i]|>b_{j+1}$, then we remove entry $r[i]$ from the DP table by setting $r[i]=-\infty$, and leave $T^-[i],T^+[i]$ undefined.
\end{itemize}
This clearly satisfies the requirement of hint size: the positive hint size is $|T^+[i]| \le b_{j+1}$, and the negative hint size is $|T^-[i]| = \big \lvert S^-\big [z[i]\big ]\big \rvert \le b_j$.
From \cref{eqn:temphintsize}, we also know that we did not remove the entry $r[\hat{\imath}]$ corresponding to our $(J_{\le j}^+\cup J_{\le j-1}^-)$-optimal solution $(\widehat A'', B'')$ defined earlier 
(this is because $|I_{\{w\}}\cap J_j^+| = x[\hat{\imath}]_w = y_w$ would imply $I_{\{w\}}\cap J_j^+ \subseteq Y \subseteq A'$).

Finally we show that $T^+[\hat{\imath}]$ and $T^-[\hat{\imath}]$ satisfy \cref{item:hinteditem2} of \cref{defn:hinteddptable} for the DP table entry $r[\hat{\imath}]$.
Suppose for contradiction that there is an $I_{\caW_1}$-optimal partial solution $(A^*, B^*)$ whose restriction  
$(\widehat A^{**}, B^{**})$ where $\widehat A^{**} := A^* \cap J_{\le j}^+, B^{**} := B^* \cap J_{\le j-1}^-$  has weight $W(\widehat A^{**}) - W(B^{**}) = \hat{\imath}$,  such that either 
$A^* \setminus J_{\le j}^+ \not \subseteq I_{T^+[\hat{\imath}]}$ or  $B^* \setminus J_{\le j-1}^-\not \subseteq I_{T^-[\hat{\imath}]}$.
      Recall that $(A''\cup Y, B'')$ (where $A''\subseteq J_{\le j-1}^+, Y\subseteq J_j^+, B''\subseteq J_{\le j-1}^-$) has the same weight $\hat{\imath}$ and profit $r[\hat{\imath}]$ as $(\widehat A^{**}, B^{**})$ does, so  
    \begin{equation}
        \label{temp:newsol}
     (\tilde A^*, \tilde B^*):=  \big ((A''\cup Y) \cup (A^*\setminus \widehat A^{**}),\,  B'' \cup (B^*\setminus B^{**})\big ) 
    \end{equation} 
         is also an $I_{\caW_1}$-optimal partial solution with the same weight and profit as $(A^*,B^*)$.
Now consider two cases:
\begin{itemize}
    \item Case $A^* \setminus J_{\le j}^+ \not \subseteq I_{T^+[\hat{\imath}]}$:
        
    Let $a \in A^* \setminus J_{\le j}^+$ and $a\notin I_{T^+[\hat{\imath}]}$.
      By
     $a\notin I_{T^+[\hat{\imath}]}$ and the
     definition of $T^+[\hat{\imath}]$, we know either $|I_{\{w_a\}} \cap J_{j}^+| > x[\hat{\imath}]_{w_a}$ or $I_{\{w_a\}} \cap J_{j+1}^+= \varnothing$.
      The latter is impossible given the existence of $a\in A^* \setminus J_{\le j}^+ \subseteq (J_{j+1}^+\cup \dots \cup J_{k}^+)$ (by \cref{defn:rankpartitioning}), so we must have $|I_{\{w_a\}} \cap J_{j}^+|> x[\hat{\imath}]_{w_a}$, which means some item $a' \in I_{\{w_a\}} \cap J_{j}^+$ is not included in the set $Y\subseteq J_{j}^+$ that encodes $\vec x[\hat{\imath}]$. 
      Hence, the $I_{\caW_1}$-optimal partial solution $(\tilde A^*,\tilde B^*)$ in \cref{temp:newsol} contains $a$ but not $a'$.
      Since $\rank(a')<\rank(a)$ and $w_{a'}=w_a$, this violates the prefix property (\cref{lem:prefix}) and hence contradicts the $I_{\caW_1}$-optimality of $(\tilde A^*,\tilde B^*)$.

    \item Case $B^* \setminus J_{\le j-1}^-\not \subseteq I_{T^-[\hat{\imath}]}$:

    By definition, $T^-[\hat{\imath}] = S^-\big [z[\hat{\imath}]\big ] = S^-[z'']$ where $z'' = W(A'')-W(B'')$.
    Note that the $I_{\caW_1}$-optimal partial solution $(\tilde A^*,\tilde B^*)$ in \cref{temp:newsol} satisfies $W(\tilde A^*\cap J^+_{\le j-1}) - W(\tilde B^* \cap J^-_{\le j-1}) = W(A'') - W(B'') = z''$, so by the assumption that $q[\,]$ (with hints $S^-[\,], S^+[\,]$) is hinted-$(J_{\le j-1}^+\cup J_{\le j-1}^-)$-optimal, we must have  $\tilde B^* \setminus J_{\le j-1}^- \subseteq I_{S^-[z'']}$ by \cref{defn:hinteddptable} (\cref{item:hinteditem2}). However, $ \tilde B^* \setminus J_{\le j-1}^- =  B^* \setminus J_{\le j-1}^- \not \subseteq I_{T^-[\hat{\imath}]} = I_{S^-[z'']} $, a contradiction.
\end{itemize}
Hence, we have verified that $T^+[\hat{\imath}]$ and $T^-[\hat{\imath}]$ satisfy \cref{item:hinteditem2} of \cref{defn:hinteddptable} for the DP table entry $r[\hat{\imath}]$, so $r[\,]$ (with hints $T^+[\,], T^-[\,]$) is indeed a hinted-$(J_{\le j}^+\cup J_{\le j-1}^-)$-optimal DP table.
This finishes the proof of the correctness of our algorithm.

The time complexity of applying \cref{lem:prob3largeb} is $O(L_j b_j\log^2(L_jb_j))$, and the time complexity of preprocessing functions $Q_w(\cdot)$ is 
$O(|J_{j}^+| + |\caW_1|\ww\log \ww)$. Since $|\caW_1| \le O((\ww\log \ww)^{1/2})$ and $L_jb_j =  \Theta(\ww^2 \log \ww)$ (by \cref{eqn:defbj,eqn:defmj}), the total running time is $O(L_{j}b_j\cdot \log^2(L_jb_j) + |J^+_{j}|)$.
\end{proof}

\section{Algorithm for $\textsc{HintedKnapsackExtend}^+$}
\label{sec:propagation}

In this section we describe our algorithm for the $\textsc{HintedKnapsackExtend}^+$ problem (\cref{prob:prob3}), proving \cref{lem:prob3largeb}. In \cref{lem:singleton}, we solve the special case where the hints are singleton sets. Then in \cref{subsec:help}, we provide several helper lemmas that allow us to decompose an instance into multiple instances with smaller hint sets. Finally in \cref{subsec:color-coding} we put the pieces together to solve the general case.

\subsection{The base case with singleton hint sets}
\label{subsec:singleton}
The following lemma is the most interesting building block of our algorithm for \cref{prob:prob3}.%

\begin{algorithm}
\DontPrintSemicolon
\caption{Solving $\textsc{HintedKnapsackExtend}^+$ with singleton hint sets}
\label{alg:knapsack-batch-update}
\KwIn{$q[-L\dd L]$ and $S[-L\dd L]$, where $S[i] \subseteq [\ww], |S[i]|\le 1, q[i]\in \Z\cup \{-\infty\}$ for all $i$}
\KwOut{$(\vec x[-L\dd L],  z[-L\dd L], r[-L\dd L])$ as a solution to \cref{prob:prob3}}
$\textsc{SMAWKAndScan}(q[-L\dd L],S[-L\dd L])\colon$\\
\Begin{
    \tcc{Stage 1: use SMAWK to find all candidate updates $q[j] + Q_w(\frac{i-j}{w})$ where $w\in S[j]$, expressed as difference-$w$ APs consisting of indices $i$}
    Initialize $\caP \gets \varnothing$\\
    \For{$w\in [\ww]$ and $c\in \{0,1,\dots,w-1\}$\label{line:forloop}} {
        $J := \{j : w\in S[j]\text{ and } j\equiv c\pmod{w}, -L\le j\le L\}$ \label{line:defnJ}\\
        $I := \{i :  i\equiv c\pmod{w}, -L\le i \le L\}$\\
        Run SMAWK (\cref{thm:smawk}) on matrix $A_{I\times J}$ defined as $A[i,j]:= q[j] + Q_w\left (\frac{i-j}{w} \right )$. \label{line:smawk}\\
        \For{$j\in J$}{
        Suppose SMAWK returned the  AP  $P_j \subseteq I$ of difference $w$, such that for every $i\in P_j$, $j = \arg \max_{j'\in J}A[i,j']$\\
        $P_j \gets P_j \cap  \{i\in \Z : i>j\}$ \label{line:removej}\tcp*[r]{focus on candidate updates where $\frac{i-j}{w}$ is a positive integer} 
        Suppose $P_j = \{c+k w, c+(k+1)w,\dots,c+\ell w\}$, and insert $(j;c,w,k,\ell)$ into $\caP$\\ %
        }
    }
    \tcc{Stage 2: combine all candidate updates by a linear scan from left to right, extending winning APs and discarding losing APs}
    Initialize empty buckets $B[-L],B[-L+1],\dots,B[L]$\\
    \For{$(j;c,w,k,\ell)\in \caP$}{
        \label{line:apbegins}
        Insert $(j;c,w,k,\ell)$ into bucket $B[c+k w]$ \tcp*[r]{insert to the bucket indexed by the beginning element of the AP}
    }
\For{$i \gets -L,\dots,L$ \label{line:scan}}{
    $r[i] \gets q[i], z[i] \gets i, \vec{x}[i] \gets \vec 0$.\label{line:trivial} \tcp*[r]{the trivial solution for $i$}
    \If{$B[i] \neq \varnothing$}{
    Pick $(j;c,w,k,\ell)\in B[i]$ that maximizes $q[j] + Q_w\left (\frac{i-j}{w} \right )$ \label{line:checkbucket}\\
    \If{$q[j] + Q_w\left (\frac{i-j}{w} \right ) > r[i]$\label{line:ifbetter}}{
     $r[i] \gets q[j] + Q_w\left (\frac{i-j}{w} \right ), z[i] \gets j, \vec{x}[i] \gets \frac{i-j}{w}\vec e_w$. \label{line:updatei}\tcp*[r]{solution for $i$}
    }
     \If{$i+w\le c+\ell w$}{ Insert $(j;c,w,k,\ell)$ into bucket $B[i+w]$ \label{line:insertanother}\tcp*[r]{extend this winning AP by one step, and all other APs in the bucket $B[i]$ are discarded} }
    }
}
\Return{$(\vec x[-L\dd L],  z[-L\dd L], r[-L\dd L])$}
}
\end{algorithm}

\begin{lemma}
    \label{lem:singleton}
    
   $\textsc{HintedKnapsackExtend}^+$ (\cref{prob:prob3}) with  $|S[i]|\le 1$ for all $-L\le i \le L $ can be solved deterministically in $O(L\log L)$ time.
   
   More precisely, the algorithm runs in $O(L + L_1\log L)$ time, where $L_1 = \{ -L\le i\le L : S[i]\neq \varnothing\}$.
\end{lemma}

The pseudocode of our algorithm for \cref{lem:singleton} is given in \cref{alg:knapsack-batch-update}. Here we first provide an overview.  \cref{alg:knapsack-batch-update} contains two stages:
\begin{itemize}
    \item In the first stage, we enumerate $w\in [\ww]$ and $c \pmod{w}$, and collect  indices $j \equiv c\pmod{w}$
such that $w\in S[j]$. 
    Then we try to extend from these collected indices $j$ by adding integer multiples of $w$ (which does not interfere with other congruence classes modulo $w$): using SMAWK algorithm \cite{smawk}, for every $i\equiv c\pmod{w}$, find $j$ among the collected indices to maximize $q[j]+Q_w(\frac{i-j}{w})$.
This is the same idea as in the proof of the standard batch-update \cref{lem:batch} (used in e.g., \cite{KellererP04,Chan18a,AxiotisT19,icalp21}). 
However, in our scenario with small sets $S[j]$, the number of collected indices $j$ is usually sublinear in the array size $L$,  so in order to save time we need to let SMAWK return a compact output representation, described as several arithmetic progressions (APs) with difference $w$, where each AP contains the indices $i$ that have the same maximizer $j$.

\item The second stage is to combine all the APs found in the first stage, and update them onto a single DP array. Ideally, we would like to take the entry-wise maximum over all the APs, that is, for each $i$ we would like to maximize $q[j]+Q_w(\frac{i-j}{w})$ over all APs containing $i$, where $w$ is the difference of the AP and $j$ is the maximizer associated to that AP. Unfortunately, the total length of these APs could be much larger than the array size $L$, which would prevent us from getting an $\tilde O(L)$ time algorithm.  
    To overcome this challenge, the idea here is to crucially use the relaxation in the definition of \cref{prob:prob3}, so that we can skip a lot of computation based on the concavity of $Q_w(\cdot)$. We perform a linear scan from left to right, and along the way we discard many APs that cannot contribute to any useful answers. In this way we can get the time complexity down to near-linear.
\end{itemize}

\begin{proof}[Proof of \cref{lem:singleton}]
The algorithm is given in \cref{alg:knapsack-batch-update}.  
\paragraph*{Time complexity.} We first analyze the time complexities of the two stages of \cref{alg:knapsack-batch-update}.
\begin{itemize}
    \item 
 The first stage contains a \textbf{for} loop over $w\in [\ww]$ and $c\in \{0,1\dots,w-1\}$ (\cref{line:forloop}), but we actually only need to execute the loop iterations such that the index set $J := \{j : w\in S[j]\text{ and } j\equiv c\pmod{w}, -L\le j\le L\}$ (defined at \cref{line:defnJ}) is non-empty. Since $|S[j]|\le 1$ for all $j$, these sets $J$ over all $(w,c)$ form a partition of the size-$L_1$ set $\{-L\le j \le L: S[j]\neq \varnothing\}$, and can be prepared efficiently. Then, for each of these sets $J$, at \cref{line:smawk} we run SMAWK algorithm (\cref{thm:smawk}) to find all row maxima (with compact output representation) of an $O(1 + L/w)\times |J|$ matrix in $O(|J|\log L)$ time. 
 The output of SMAWK is represented as $|J|$ intervals on the row indices of this matrix, which correspond to  $|J|$ APs of difference $w$. These $|J|$ APs are then added into $\caP$. 
 Thus, in the end of the first stage, set $\caP$ contains at most $\sum_J |J| \le L_1\le  2L+1$ APs (each AP only takes $O(1)$ words to describe), and the total running time of this stage is $O(\sum_J |J|\log L) = O(L_1\log L)$.
 \item In the second stage, we initialize $(2L+1)$ buckets $B[-L\dd L]$, and first insert each AP from $\caP$ into a bucket (\cref{line:apbegins}). %
  Then we do a scan $i\gets -L,\dots,L$ (\cref{line:scan}), where for each $i$ we examine all APs in the bucket $B[i]$ at \cref{line:checkbucket}, and then copy at most one winning AP from this bucket into another bucket (\cref{line:insertanother}). Hence, in total we only ever inserted at most $|\caP|+(2L+1) = O(L)$ APs into the buckets. So the second stage takes $O(L)$ overall time.
\end{itemize}
Hence the total time complexity of \cref{alg:knapsack-batch-update} is $O(L_1\log L + L) \le O(L\log L)$. 

\paragraph*{Correctness.}
It remains to prove that the return values
$(\vec x[i],  z[i], r[i])$ correctly solve \cref{prob:prob3}. Fix any $i\in \{-L,\dots,L\}$, and let $(\vec x^*[i],  z^*[i], r^*[i])$ be an maximizer of \cref{eqn:tomaximize} (subject to \cref{eqn:subjectto}).
If $\vec x^*[i] = \vec 0$, then it is the trivial solution, which cannot be better than our solution, due to \cref{line:trivial,line:ifbetter}. So in the following we assume $\lvert\supp(\vec x^*[i])\rvert\ge 1$, which means $i>z^*[i]$.
If the support containment condition $\supp(\vec x^*[i]) \subseteq S\big [z^*[i]\big ]$ (\cref{eqn:supportcontain}) is violated, then by definition of \cref{prob:prob3} we are not required to find a maximizer for $i$.
Hence, we can assume $\supp(\vec x^*[i]) \subseteq S\big [z^*[i]\big ]$ holds.
Since $\big \lvert S\big [z^*[i]\big ]\big \rvert\le 1 \le |\supp(\vec x^*[i])|$,  we assume $\supp(\vec x^*[i]) = S\big [z^*[i]\big ] = \{w^*\}$.

In the \textbf{for} loop iteration of the first stage where $w=w^*$ and $c = i\bmod w$, we have $z^*[i] \in J$ and $i\in I$.
The input matrix $A_{I\times J}$ to SMAWK encodes the objective values of extending from $j$ by adding multiples of $w^*$;
in particular, $A[i,z^*[i]]$ equals our optimal objective $r^*[i]= q\big [z^*[i]\big ] + Q_{w^*}\left ( \frac{i-z^*[i]}{w^*}\right )$. So SMAWK correctly returns an AP $P_{z^*[i]} = \{c+k w^*, c+(k+1)w^*,\dots,c+\ell w^*\}$ that contains $i$
(unless there is  a tie $A[i,z^*[i]] = A[i,j]$ for some other $j\in J$, and $i$ ends up in the AP $P_j$, but in this case we could have started the proof with $(\vec x^*[i], z^*[i])$ being this alternative maximizer $z^*[i]\gets j$ and $\vec x^*[i] \gets \frac{i-z^*[i]}{w^*} \vec{e}_{w^*}$). Since $i>z^*[i]$, we know $i$ is not removed from $P_{z^*[i]}$ at \cref{line:removej}.
This AP $P_{z^*[i]}$ containing $i$ is then added to $\caP$.

In the second stage, each AP in $\caP$ starts in the bucket indexed by the leftmost element of this AP (\cref{line:apbegins}), and during the left-to-right linear scan this AP may win over others in its current bucket (at \cref{line:checkbucket}) and gets advanced to the bucket corresponding to its next element in the AP (at \cref{line:insertanother}), or it may lose at \cref{line:checkbucket} and be discarded.
(Note that any AP can only appear in buckets whose indices belong to this AP.)
Our goal is to show that the AP $P_{z^*[i]}$ can survive the competitions and arrive in bucket $B[i]$, so that it can successfully update the answer for $i$  at \cref{line:updatei}.

Suppose for contradiction that $P_{z^*[i]}$ lost to some other AP $P'_{j'}$ at \cref{line:checkbucket} when they were both in bucket $B[i_0]$ (for some $i_0<i$). %
Suppose this AP $P'_{j'}$ has common difference $w'$, and corresponds to the objective value $q[j'] + Q_{w'}\left ( \frac{i'-j'}{w'}\right )$ for $i'\in P'_{j'}$. Note that $i_0\in P_{j'}'$ satisfies $i_0>j'$ due to \cref{line:removej}.
Note that $w' \neq w^*$ must hold, since two APs produced in stage 1 with the same common difference cannot intersect (because SMAWK (\cref{thm:smawk}) returns disjoint intervals), and hence cannot appear in the same bucket $B[i_0]$.
 Now we consider an alternative solution for index $i$ defined as 
\[ ( {\vec x'}, j'):=\big(\tfrac{i_0-j'}{w'}\vec e_{w'} + \tfrac{i-i_0}{w^*}\vec e_{w^*}\,,\, j'\big ). \]
Note that $j'+\sum_{w\in [\ww]}w\cdot x'_w = i$, and it has objective value
\begin{align*}
& q[j'] + \sum_{w\in [\ww]} Q_w(x'_w) \\
 = \  &    q[j'] + Q_{w'}\left ( \tfrac{i_0-j'}{w'}\right ) + Q_{w^*}\left ( \tfrac{i-i_0}{w^*}\right )\\
 \ge \ & q\left [z^*[i]\right ] + Q_{w^*}\left ( \tfrac{i_0-z^*[i]}{w^*}\right ) + Q_{w^*}\left ( \tfrac{i-i_0}{w^*}\right ) \tag{since $P_{j'}$ wins over $P_{z^*[i]}$ in bucket $B[i_0]$}\\
 \ge  \ & q\left [z^*[i]\right ] + Q_{w^*}\left ( \tfrac{i_0-z^*[i]}{w^*} + \tfrac{i-i_0}{w^*}\right ) + 0 \tag{by  concavity of $Q_{w^*}(\cdot)$}\\
 = \ & r^*[i].
\end{align*}
Now there are two cases:
\begin{itemize}
    \item $q[j'] + \sum_{w\in [\ww]} Q_w(x'_w) > r^*[i]$.

This contradicts the assumption  that $r^*[i]$ is the optimal objective value for index $i$.
        \item 
$q[j'] + \sum_{w\in [\ww]} Q_w(x'_w) = r^*[i]$.

Then, $(\vec x', j')$ is also a maximizer for index $i$, but it has support size $|\supp(\vec x')|=2$ due to $i>i_0>j'$ and $w'\neq w^*$, and hence violates the support containment condition $\supp(\vec x') \subseteq S[j']$ (\cref{eqn:supportcontain}). By definition of \cref{prob:prob3}, we are not required to find a maximizer for index $i$.
\end{itemize}

Hence, we have shown that the AP $P_{z^*[i]}$ can arrive in bucket $B[i]$. This finishes the proof that \cref{alg:knapsack-batch-update} correctly solves $\textsc{HintedKnapsackExtend}^+$ for index $i$.
\end{proof}

\subsection{Helper lemmas for problem decomposition}
\label{subsec:help}
In this section, we show several helper lemmas for the 
$\textsc{HintedKnapsackExtend}^+$ problem (\cref{prob:prob3}). 
To get some intuition, one may first consider \cref{prob:prob3} without the relaxation based on hints, which is a standard dynamic programming problem that obeys some kind of composition rule: namely, if we update a DP table $q[\,]$ with items from $U_1\cup U_2$ (for some disjoint $U_1$ and $U_2$), it should have the same effect as first updating $q[\, ]$ with $U_1$ to obtain an intermediate DP table, and then updating this intermediate table with $U_2$.
The goal of this section is to formulate and prove analogous composition properties for the $\textsc{HintedKnapsackExtend}^+$ problem, which will be useful for our decomposition-based algorithms to be described later in \cref{subsec:color-coding}.
The proofs of these properties are quite mechanical and are similar to the arguments in the proof of \cref{lem:firststage} in \cref{sec:3.3}, and can be safely skipped at first read.

Using the notations from \cref{prob:prob3}, we denote an instance of $\textsc{HintedKnapsackExtend}^+$ as 
\[K = \big (U,\{Q_w\}_{w\in U}, S[-L\dd L] , q[-L\dd L]\big ),\]
where $S[i] \subseteq U$ for all $i$.  And we denote a solution to $K$ as 
\[ Y = (\vec x[-L\dd L], z[-L\dd L], r[-L\dd L]).\]
where $z[i]  = i-  \sum_{w\in U}w\cdot x[i]_w$ by \cref{eqn:subjectto}, and 
\[ r[i]:= q\big [z[i]\big ] + \sum_{w\in V}Q_w(x[i]_w)\]
is the objective value.
We say a solution \emph{correctly solves} $K$ if it satisfies the definition of \cref{prob:prob3} (with the relaxation based on hints $S[\,]$).

Now we define several operations involving the instance $K$. In the following we omit the array index range $[-L\dd L]$ for brevity.
\begin{definition}[Restriction]
The \emph{restriction} of instance $K$ to a set $V\subseteq U$ is defined as the instance
\[ K\lvert_{V}:= \big (V,\{Q_w\}_{w\in V}, S_V[\, ] , q[\, ]\big )\]
where $S_V[i]:= S[i] \cap V$.
\end{definition}

\begin{definition}[Updating]
Suppose $Y_V = (\vec x[\, ], z [\, ], r[\, ])$ is a solution to $K\lvert_{V}$, then we define the following updated instance
\[ K^{(V\gets Y_V)} := (U\setminus V, \{Q_w\}_{w\in U\setminus V}, S'[\,],q'[\,])\]
where 
\begin{equation}
    \label{eqn:nextsets}
 S'[i]:= S\big [z[i]\big] \setminus V,
\end{equation}
and
\begin{equation*}
 q'[i]:= r[i].
\end{equation*}
\end{definition}

\begin{definition}[Composition]
    Let $V,V'\subseteq U, V\cap V'= \varnothing$.
Suppose $Y_V = (\vec x[\, ], z [\, ], r[\, ])$ is a solution to $K\lvert_{V}$, and  $Y_{V'}= ({\vec x}'[\, ], z' [\, ], r'[\, ])$ is a solution to $K^{(V\gets Y_V)}\lvert_{V'}$. 
We define the following composition of solutions,
\[ Y_{V'} \circ Y_V := ({\vec x}''[\, ], z'' [\, ], r'[\, ]) \]
as a solution to $K\lvert_{V\cup V'}$, where 
\[ z''[i]:= z\big [z'[i]\big],\]
\[ {\vec x}''[i]:= {\vec x}'[i] + {\vec x}\big [z'[i]\big ].\]
Note that $\circ$ is associative.%
\end{definition}

Now we are ready to state the composition lemma for $\textsc{HintedKnapsackExtend}^+$.
\begin{lemma}[Composition lemma for $\textsc{HintedKnapsackExtend}^+$]
    \label{lem:composeweak}
    Let $V,V'\subseteq U, V\cap V'= \varnothing$.
    If $Y_V$  correctly solves $K\lvert_V$, and $Y_{V'}$ correctly solves $K^{(V\gets Y_V)}\lvert_{V'}$, then $Y_{V'} \circ Y_V$ correctly solves $K\lvert_{V\cup V'}$.
\end{lemma}
\begin{proof}
Recall $K = \big (U,\{Q_w\}_{w\in U}, S[\,] , q[\,]\big )$.
Denote the solutions $Y_V = (\vec x[\, ], z [\, ], r[\, ])$, $Y_{V'}= ({\vec x}'[\, ], z' [\, ], r'[\, ])$, and $Y_{V'} \circ Y_V := ({\vec x}''[\, ], z'' [\, ], r''[\, ])$.

Suppose for contradiction that for some $i$, ${\vec x}''[i]$ is incorrect for the instance $K\lvert_{V\cup V'}$.
By definition of $\textsc{HintedKnapsackExtend}^+$, this means that all maximizers $(\bar{\vec x}, \bar{z})$ (where $\bar{\vec x}\in \Z_{\ge 0}^{V\cup V'}$) of the objective value
 \[ \bar r = q[\bar z]  + \sum_{w\in V\cup V'}Q_w(\bar x_w)\]
 (subject to $\bar z + \sum_{w\in V\cup V'}w\cdot \bar x_w = i$)
should satisfy the support containment condition, 
\[\supp(\bar{\vec x}) \subseteq S[\bar z]. \]
Take any such maximizer $(\bar{\vec x}, \bar{z})$, and write $\bar {\vec x} = \bar {\vec x}_{V} + \bar {\vec x}_{V'}$ such that $\supp(\bar {\vec x}_{V})\subseteq V$ and $\supp(\bar {\vec x}_{V'})\subseteq V'$. Let 
\[ i_V:= \bar z + \sum_{w\in V}w\cdot \bar x_w\]
and
\[r_{V}:= q[\bar z] + \sum_{w\in V}Q_w(\bar x_w).\]
Note that in instance $K\lvert_V$, $(\bar{\vec x}_{V}, \bar z, r_{V})$ should be a valid solution for $i_V$ with objective value $r_V$. We now compare it with $r[i_V]$ from the correct solution $Y_V$ to the instance $K\lvert_{V}$, and we claim that $r[i_V]=r_V$ must hold. Otherwise, by the  definition of the $\textsc{HintedKnapsackExtend}^+$ instance $K\lvert_{V}$,  there can only be two possibilities:
\begin{itemize}
    \item Case 1: $(\bar{\vec x}_{V}, \bar z, r_{V})$ is not a maximizer solution for index $i_V$ in $K\lvert_{V}$.
    
    This means there is some solution $(\vec x^\star,  z^\star, r^\star)$ for index $i_V$ in $K\lvert_{V}$ that achieves a higher objective $r^\star > r_V$. 
    We will use an exchange argument to derive contradiction:  Consider  the solution $(\vec x^\star + \bar {\vec x}_{V'}, z^\star)$ to the instance $K\lvert_{V\cup V'}$. It has the same total weight $z^\star + \sum_{w\in V} w\cdot x^\star_w + \sum_{w\in V'}w\cdot \bar {x}_{w} = i_V +  \sum_{w\in V'}w\cdot \bar {x}_{w} = i$, but has a higher objective value 
    $ r^\star + \sum_{w\in V'}Q_w(\bar x_w) > r_V + \sum_{w\in V'}Q_w(\bar x_w) = \bar r$, contradicting the assumption that $(\bar{\vec x}, \bar{z})$ is a maximizer for $i$ in instance $K\lvert_{V\cup V'}$.

    \item Case 2: $(\bar{\vec x}_{V}, \bar z, r_{V})$ is a maximizer solution for index $i_V$ in $K\lvert_{V}$, but there is another maximizer solution $(\vec x^\star,  z^\star, r_V)$ for index $i_V$ in $K\lvert_{V}$ that does not satisfy the support containment condition $\supp(\vec x^\star)\subseteq S[z^\star]$ for instance $K\lvert_{V}$.
    
    In this case, again consider the solution $(\vec x^\star + \bar {\vec x}_{V'}, z^\star)$ to the instance $K\lvert_{V\cup V'}$. This time it has  the same objective value $\bar r$, so it is a maximizer for index $i$ in the instance $K\lvert_{V\cup V'}$. However, since $\supp(\vec x^\star)\nsubseteq S[z^\star]$, we know $\supp(\vec x^\star+ \bar {\vec x}_{V'})\nsubseteq S[z^\star]$, and hence it violates the support containment condition in $K\lvert_{V\cup V'}$, contradicting our assumption that all maximizers to index $i$ in instance $K\lvert_{V\cup V'}$ satisfy the support containment condition.
\end{itemize}
Hence we have established that $r[i_V]=r_V$ must hold.

Now we look at the second instance, $K^{(V\gets Y_V)}\lvert_{V'} = \big (V',\{Q_w\}_{w\in V'}, S'[\,] , q'[\,]\big )$. Note that we have $q'[i_V]=r[i_V]=r_V$ by definition. Hence, $(\bar{\vec x}_{V'},i_V,\bar r)$ should be a valid solution for index $i$ in instance $K^{(V\gets Y_V)}\lvert_{V'}$, with objective value
\[ q'[i_V] + \sum_{w\in V'}Q_w(\bar x_w) = r_V + \sum_{w\in V'}Q_w(\bar x_w) = \bar r.\]
Now we claim that $r'[i] = \bar r$ (recall that $r'[\cdot]$ denotes objective values achieved by the solution $Y_{V'}$) must hold. Otherwise, by the definition of the 
$\textsc{HintedKnapsackExtend}^+$ instance $K^{(V\gets Y_V)}\lvert_{V'}$,  there can only be two possibilities:
\begin{itemize}
    \item Case 1: $(\bar{\vec x}_{V'},i_V,\bar r)$ is not a maximizer solution for index $i$ in $K^{(V\gets Y_V)}\lvert_{V'}$.

    This means there is some solution $(\vec x^*,  z^*,  r^*)$ for index $i$ in $K^{(V\gets Y_V)}\lvert_{V'}$ that achieves a higher objective $r^*>\bar r$. 
    Then, in instance $K\lvert_{V\cup V'}$, the solution $(\vec x^* + \vec x[z^*], z[z^*])$ has total weight $z[z^*] +   \sum_{w\in V}w\cdot x[z^*]_w + \sum_{w\in V'}w\cdot x^*_w = z^* +  \sum_{w\in V'}w\cdot x^*_w = i$ and total objective value 
$q[z[z^*]] +   \sum_{w\in V}Q_w(x[z^*]_w) + \sum_{w\in V'}Q_w( x^*_w) = r[z^*] +  \sum_{w\in V'}Q_w(x^*_w) =q'[z^*] +  \sum_{w\in V'}Q_w(x^*_w) = r^*>\bar r$, contradicting the assumption that $(\bar{\vec x}, \bar{z})$ is a maximizer for $i$ in instance $K\lvert_{V\cup V'}$.

    \item Case 2: $(\bar{\vec x}_{V'},i_V,\bar r)$ is a maximizer solution for index $i$ in $K^{(V\gets Y_V)}\lvert_{V'}$, but there is another maximizer solution
    $(\vec x^*,  z^*, \bar r)$ for index $i$ in $K^{(V\gets Y_V)}\lvert_{V'}$ that does not satisfy the support containment condition $\supp(\vec x^*)\subseteq S'[z^*]$. 

Similar to Case 1, we again consider the solution $(\vec x^* + \vec x[z^*], z[z^*])$ in instance $K\lvert_{V\cup V'}$, which has total weight $i$ and total objective value $\bar r$. So it is an alternative maximizer to index $i$ in instance 
$K\lvert_{V\cup V'}$. 
Since $S'[z^*]= S[z[z^*]]\setminus V$ by definition (\cref{eqn:nextsets}), we have $\supp(\vec x^*)\not \subseteq S'[z^*]=S[z[z^*]]\setminus V$, which means $\supp(\vec x^*)\not \subseteq S[z[z^*]]$, and thus $\supp(\vec x^* + \vec x[z^*])\not \subseteq S[z[z^*]]$.
This contradicts our assumption that all maximizers to index $i$ in instance $K|_{V\cup V'}$ satisfy the support containment condition.
\end{itemize}
Hence, we must have $r'[i] = \bar r$. This means that we indeed have found a maximizer to index $i$ after we compose the solutions $Y_V$ and $Y_{V'}$.  So our solution for index $i$ is actually correct for the instance $K\lvert_{V\cup V'}$. 
This finishes the proof that $Y_{V'} \circ Y_V$ correctly solves $K\lvert_{V\cup V'}$.
\end{proof}

\cref{lem:composeweak} allows us to decompose an instance by partitioning the set $U\subseteq [\ww]$. In addition to this, we also need another way to decompose an instance, which in some sense allows us to partition the array indices $[-L\dd L]$. First, we define the entry-wise maximum of two instances.
\begin{definition}[Entry-wise maximum]
    \label{defn:entrywisemax}
Given $\textsc{HintedKnapsackExtend}^+$ instances $K = \big (U,$ $\{Q_w\}_{w\in U}, S[\, ] , q[\,]\big )$ and 
$K' = \big (U,\{Q_w\}_{w\in U}, S'[\, ] , q'[\,]\big )$,
we define the following instance,
\[(U,\{Q_w\}_{w\in U}, S''[\, ] , q''[\,]\big ), \]
where 
\[ (S''[i],q''[i]):= \begin{cases}
    (S[i],q[i]) & q[i]>q'[i],\\
    (S'[i],q'[i]) & q[i]<q'[i],\\
    (S[i] \cap S'[i],q[i]) & q[i]=q'[i],\\
\end{cases}\]
and denote this instance by $\max(K,K')$.
We naturally extend this definition to the entry-wise maximum of possibly more than two instances. 

We also define the entry-wise maximum of two solutions
$Y = (\vec x[\, ], z [\, ], r[\, ]), Y' = (\vec x'[\, ], z' [\, ], r'[\, ])$, by $\max(Y,Y') = (\vec x''[\, ], z'' [\, ], r''[\, ])$, where 
\[ (\vec x''[i],z''[i]):= \begin{cases}
   (\vec x[i],z[i]) & r[i]>r'[i],\\
   (\vec x'[i],z'[i]) & \text{otherwise,}\\
\end{cases}\]
and objective $r''[i]$ can be uniquely determined from $(\vec x''[i],z''[i])$. Note that $r''[i] \ge \max(r[i],r'[i])$ obviously holds.
\end{definition}

Naturally, we have the following lemma for $\textsc{HintedKnapsackExtend}^+$.
\begin{lemma}[Entry-wise maximum lemma]
    \label{lemma:instancemax}
    If $Y$ correctly solves $K$, and $Y'$ correctly solves $K'$, then $\max(Y,Y')$ correctly solves $\max(K,K')$.
\end{lemma}
\begin{proof}
Denote $K = \big (U,\{Q_w\}_{w\in U}, S[\,] , q[\,]\big )$, $K' = \big (U,\{Q_w\}_{w\in U}, S'[\,] , q'[\,]\big )$, $Y = (\vec x[\, ], z [\, ], r[\, ]), Y' = (\vec x'[\, ], z' [\, ], r'[\, ])$. Let $\max(K,K') = \big (U,\{Q_w\}_{w\in U}, S''[\,] , q''[\,]\big )$ and $\max (Y,Y') = (\vec x''[\, ], z'' [\, ], r''[\, ])$.
    
Suppose for contradiction that for some $i$, ${\vec x}''[i]$ is incorrect for the instance $\max(K,K')$.
By definition of $\textsc{HintedKnapsackExtend}^+$, this means that all maximizers $(\bar{\vec x}, \bar{z})$ (where $\bar{\vec x}\in \Z_{\ge 0}^{U}$) of the objective value
 \[ \bar r = q''[\bar z]  + \sum_{w\in U}Q_w(\bar x_w)\]
 (subject to $\bar z + \sum_{w\in U}w\cdot \bar x_w = i$)
should satisfy the support containment condition, 
\[\supp(\bar{\vec x}) \subseteq S''[\bar z]. \]
Take any such maximizer $(\bar{\vec x}, \bar{z})$.
Since $q''[\bar z] = \max(q[\bar z], q'[\bar z])$, without loss of generality we assume $q''[\bar z] = q[\bar z]$ by symmetry.
Then, in instance $K$, $(\bar{\vec x},\bar z)$ is a valid solution for index $i$, achieving objective value $q[\bar z] + \sum_{w\in U}Q_w(\bar x_w) = \bar r$. If in solution $Y$, the found objective $r[i]$ satisfies $r[i]\ge \bar r$, then in the entry-wise maximum solution we would have $r''[i]\ge \max(r[i],r'[i])\ge \bar r$, which contradicts the assumption that $\vec x''[i]$ is incorrect. Hence, $r[i]<\bar r$ holds.  Then, since $Y$ is a correct solution to instance $K$, we know from the definition of $\textsc{HintedKnapsackExtend}^+$ instance $K$ that there can only be two possibilities:
\begin{itemize}
    \item Case 1: $(\bar{\vec x},\bar z)$ is not a maximizer solution for index $i$ in $K$.

    This would immediately mean that the actual maximum objective for index $i$ in instance $\max(K,K')$ is also greater than $\bar r$, a contradiction.

    \item Case 2: $(\bar{\vec x},\bar z)$ is a maximizer solution for index $i$ in $K$, but there is another maximizer solution $(\vec x^*, z^*)$ for index $i$ that does not satisfy the support containment condition $\supp(\vec x^\star)\subseteq S[z^\star]$ for instance $K$.

    Then, there are two cases:
    \begin{itemize}
        \item Case (2i): $q[z^*]< q'[z^*]$.

        Then, in instance $\max(K,K')$, $(\vec x^*,z^*)$  actually achieves a higher objective $\bar r+ q'[z^*]- q[z^*]$, a contradiction.
        \item Case (2ii): $q[z^*]\ge q'[z^*]$.
        
        Then, in instance $\max(K,K')$, by definition we have $S''[z^*] \subseteq S[z^*]$. Hence, $(\vec x^*,z^*)$ is a maximizer solution for index $i$ in instance $\max(K,K')$ that does not satisfy the support containment condition $\supp(\vec x^*) \subseteq S''[z^*]$, a contradiction.
    \end{itemize}
\end{itemize}
Hence,  we have reached contradictions in all cases. This means $\max(Y,Y')$ is a correct solution to $\max(K,K')$.
\end{proof}

 \subsection{Decomposing the problem via color-coding}
 \label{subsec:color-coding}

 In \cref{subsec:singleton}, we solved $\textsc{HintedKnapsackExtend}^+$ (\cref{prob:prob3}) in the base case where each hint set $S[i]$ has size at most $1$.
In this section, we extend it to the case with larger size bound $|S[i]|\le b$. This is achieved by using the color-coding technique \cite{AlonYZ95} to isolate the elements in the sets $S[i]$, and using the helper lemmas from \cref{subsec:help} to combine the solutions for different color classes.  
Here we apply the two-level color-coding approach previously used by Bringmann \cite{Bringmann17} in his near-linear time randomized subset sum algorithm.\footnote{An advantage in our scenario is that the sets $S[i]$ to be isolated are already given to us as input, so we can derandomize the color-coding technique, whereas in Bringmann's subset sum algorithm \cite{Bringmann17} the set to be isolated is the unknown solution set. Derandomizing Bringmann's algorithm is an important open problem.}

The following \cref{lem:algosmallb} gives an algorithm for $\textsc{HintedKnapsackExtend}^+$ which is
suitable for $b$ slightly larger than $1$ (for example, polylogarithmic).
\begin{lemma}[Algorithm for small $b$]
    \label{lem:algosmallb}
   $\textsc{HintedKnapsackExtend}^+$ (\cref{prob:prob3}) with  $|S[i]|\le b$ for all $-L\le i \le L $ can be solved deterministically in 
$O(Lb^2(b+\log L))$
time.
\end{lemma}
\cref{lem:algosmallb} will be proved later in this section via color-coding.
Using \cref{lem:algosmallb} as a building block, we can solve $\textsc{HintedKnapsackExtend}^+$ for even larger $b$, using another level of color-coding. 
The derandomized version of this color-coding is summarized in the following lemma.
\begin{restatable}[Deterministic balls-and-bins]{lemma}{detballs}
    \label{thm:detballs}
    Given integer $r$, and $m$ sets $S_1,S_2,\dots, S_m\subseteq [n]$ such that $|S_i|\le r\log_2(2m)$ for all $i$, there is an $O(mr\log m\log r + n\log r)$-time deterministic algorithm that finds an $r$-coloring $C\colon [n]\to [r]$, such that for every $i\in [m]$ and every color $c\in [r]$,
    \[ |\{j \in S_i: C(j)=c\}| \le O(\log m).\]
\end{restatable}
\cref{thm:detballs} follows from the results of \cite{Raghavan88,spencer1987ten} on deterministic set balancing using the method of pessimistic estimators. A proof of \cref{thm:detballs} is included in \cref{app:detcolorcode}.

Now we can prove our main \cref{lem:prob3largeb}, restated below.
\algolargeb*
\begin{proof}[Proof of \cref{lem:prob3largeb} assuming \cref{lem:algosmallb}]
    If $b\le O(\log L)$, then we can directly invoke \cref{lem:algosmallb} in $O(Lb^2(b+\log L)) \le O(Lb\log^2 L)$ time. Hence, in the following we assume $b>2\log_2(4L+2)$.
    Define parameter $r =  b/\log_2(4L+2)>2$, and use  \cref{thm:detballs} to construct in $O(Lr\log L\log r + |U|\log r)$ time a coloring $h\colon U \to [r]$, such that for every $i\in \{-L,\dots,L\}$ and every color $c\in [r]$, $S[i] \cap h^{-1}(c) \le b_1$ for some $b_1 =O(\log L) $.

Then, we iteratively apply the algorithm from \cref{lem:algosmallb} with size bound $b_1$, to solve for each color class $U_c:=h^{-1}(c)$ ($c\in [r]$). 
More precisely, let $K$ denote the input $\textsc{HintedKnapsackExtend}^+$ instance, and starting with instance $K_1:= K$, we iterate $c\gets 1,2,\dots,r$, and do the following 
(using notations from \cref{subsec:help}):
\begin{itemize}
    \item Solve the restricted instance $(K_c)\lvert_{U_c}$ using \cref{lem:algosmallb} (with size bound $b_1$), and obtain solution $Y_c$.
    \item Define the updated instance $K_{c+1}:= (K_c)^{(U_c\gets Y_c)}$.
\end{itemize}
Finally, return the composed solution $Y:= Y_{r}\circ \dots \circ Y_2\circ Y_1$. By inductively applying \cref{lem:composeweak}, we know $Y$ is a correct solution to $K$. 

We remark on the low-level implementation of the procedure described above. When we construct the new input instance for the next iteration (namely, $(K_c)^{(U_c\gets Y_c)}$), we need to prepare the new input sets $S'[\cdot ] \subseteq (U_{c+1}\cup \dots \cup U_r)$, based on the current input sets $S[\cdot ]\subseteq (U_c\cup \dots \cup U_r)$ and the current solution $z[\cdot ]$, according to \cref{eqn:nextsets}. However, each set $S[i]$ may have size as large as $b$, so it would be to slow to copy them explicitly. In contrast, note that an instance $(K_c)\lvert_{U_c}$ only asks for sets $S[i] \cap U_c$ as input, which have much smaller size $b_1 = O(\log m)$. 
So the correct implementation should be as follows: at the very beginning, each of the sets $S[-L],\dots,S[L]$ receives an integer handle, and when we need to copy the sets we actually only pass the handles. And, given the handle of a set $S[i]$ and a color class $U_c$, we can report the elements in $S[i] \cap U_c$ in $O(b_1) = O(\log m)$ time (because we can preprocess these intersections at the very beginning). In this way, the time complexity for preparing the input instances to \cref{lem:algosmallb} is no longer a bottleneck.

It remains to analyze the time complexity. There are $r$ applications of \cref{lem:algosmallb}, taking 
$O(Lb_1^2 (b_1+\log L)) \le O(L \log^3 L)$ time each, and $O(r\cdot L\log^3 L) \le O(Lb\log^2 L)$ time in total.
Combined with the time complexity of deterministic coloring at the beginning, the total time is 
$O(Lr\log L\log r + |U|\log r + Lb\log^2 L) \le O(Lb\log r + |U|\log r + Lb\log^2 L)$.
Here we can assume $|U|\le \sum_{-L\le i\le L}|S[i]|\le O(Lb)$ without loss of generality. Hence, the total time becomes $O(Lb\log r + Lb\log^2 L) \le O(Lb\log^2 (Lb))$.
\end{proof}

It remains to describe the $O(Lb^2(b+\log L))$-time algorithm claimed in \cref{lem:algosmallb}. It relies on the following birthday-paradox-type color-coding lemma, which is derandomized using a standard application of pessimistic estimators:
\begin{restatable}{lemma}{detcolorsmall}
    \label{lem:detcolorcode2}
   Given $m$ sets $S_1,S_2, \dots,S_m \in [n]$ with $|S_i|\le b$, there is a deterministic algorithm in $O(n\log m + mb^3)$ time that computes $k \le \log_2(2m)$ colorings $h_1,h_2,\dots,h_k\colon [n] \to [b^2]$, such that for every $i\in [m]$ there exists an $h_j$ that assigns distinct colors to elements of $S_i$.
\end{restatable}
We include a proof of \cref{lem:detcolorcode2} in \cref{app:detcolorcode2}. 
Now we prove \cref{lem:algosmallb}.
\begin{proof}[Proof of \cref{lem:algosmallb}]
Without loss of generality, assume $|U|\le \sum_{-L\le i\le L}|S[i]|\le O(Lb)$.
We  apply \cref{lem:detcolorcode2} to the sets $S[-L],S[-L+1]\dots,S[L]\subseteq U$, and in $O(|U|\log L + Lb^3) \le O(Lb(b^2 + \log L))$ time obtain $k=O(\log L)$ colorings $h_1,\dots,h_k\colon U \to [b^2]$ such that every set $S[i]$ is isolated by some $h_j$ (i.e., $S[i]$ receives distinct colors under coloring~$h_j$).

Now, based on the input  $\textsc{HintedKnapsackExtend}^+$ instance $K$, we define $k$ new instances $K^{(1)}, \dots, K^{(k)}$ as follows: for each $1\le j\le k$, let 
\[\caI^{(j)}= \big \{i \in  \{-L,\dots,L\}:\text{ $S[i]$ is isolated by $h_j$, but not by any $h_{j'}$ ($j'<j$)}\big \}.\] 
Then $\{\caI^{(j)}\}_{j=1}^k$  form a partition of $\{-L,\dots,L\}$.
Let instance $K^{(j)} = \big (U,\{Q_w\}_{w\in U}, S^{(j)}[\,] , q^{(j)}[\,]\big ) $ be derived from the input instance $K=\big (U,\{Q_w\}_{w\in U}, S[\,] , q[\,]\big )$ with the following modification:
\[ S^{(j)}[i] := \begin{cases} S[i] & i \in \caI^{(j)},\\ \varnothing & \text{otherwise,}\end{cases}\]
    and
\[ q^{(j)}[i] := \begin{cases} q[i] & i \in \caI^{(j)},\\ -\infty & \text{otherwise.}\end{cases}\]
Clearly, $\max(K^{(1)},\dots,K^{(k)}) = K$ (see \cref{defn:entrywisemax}), so by \cref{lemma:instancemax} it suffices to solve each $K^{(j)}$ and obtain solution $Y^{(j)}$, and finally the combined solution $Y:= \max\{Y^{(1)},\dots,Y^{(k)}\}$ is a correct solution to $K$.

Now we focus on each instance $K^{(j)}$. The sets $S^{(j)}[i]$ are isolated by the coloring $h_j$ (note that there are at most $|\caI^{(j)}|$ many non-empty sets $S^{(j)}[i]$). So we can iteratively apply the algorithm for singletons (\cref{lem:singleton}) to solve for each color class $U_c:= h_j^{-1}(c)$ ($c\in [b^2]$), in the same fashion as in the proof of \cref{lem:prob3largeb}. More precisely, starting with instance $K_1:= K^{(j)}$, we iterate $c\gets 1,2,\dots,b^2$, and do the following (using notations from \cref{subsec:help}):
\begin{itemize}
    \item Solve the restricted instance $(K_c)\lvert_{U_c}$ using \cref{lem:singleton} in $O(L + |\caI^{(j)}|\log L)$ time, and obtain solution $Y_c$.
    \item Define the updated instance $K_{c+1}:= (K_c)^{(U_c\gets Y_c)}$.
\end{itemize}
Finally, return the composed solution $Y^{(j)}:= Y_{b^2}\circ \dots \circ Y_2\circ Y_1$. By inductively applying \cref{lem:composeweak}, we know $Y^{(j)}$ is a correct solution to $K^{(j)}$. 

It remains to analyze the time complexity. Each of the  $k$ instances $K^{(j)}$ is solved by $b^2$ applications of the singleton-case algorithm (\cref{lem:singleton}) in $O(L + |\caI^{(j)}|\log L)$ time each. Hence the total time complexity over all $k$ instances is $\sum_{j=1}^k b^2\cdot O(L +|\caI^{(j)}|\log L) = O(kb^2L + b^2 \sum_{j=1}^k|\caI^{(j)}|\log L) = O(Lb^2  (k+\log L)) = O(Lb^2 \log L)$. 
Combined with the deterministic coloring step at the beginning, the overall time complexity is  $O(Lb^2\log L + Lb(b^2+\log L)) = O(Lb^2(b+\log L))$.
\end{proof}

\section*{Acknowledgements} I thank Ryan Williams and Virginia Vassilevska Williams for useful discussions.  I thank anonymous reviewers for useful comments on previous versions of this paper. 

	\bibliographystyle{alphaurl} 
	\bibliography{main}

\appendix

\section{SMAWK algorithm}
\label{app:smawk}
We review the classic result of  Aggarwal,  Klawe, Moran,  Shor, and  Wilber \cite{smawk} on finding row maxima in convex totally monotone matrices (in particular, convex Monge matrices).

We say an $m\times n$ real matrix $A$ is \emph{convex Monge} if
\begin{equation}
    \label{eqn:monge}
    A[i,j]+A[i',j'] \ge A[i,j']+A[i',j]
\end{equation} for all $i<i'$ and $j<j'$.

More generally, we consider matrices with possibly $-\infty$ entries.
Following the terminology of \cite{AggarwalK90}, a matrix $A \in (\R \cup \{-\infty\})^{m\times n}$ is called a \emph{reverse falling staircase} matrix, if the finite entries in each row form a prefix, and the finite entries in each column form a suffix. In other words, if $A[i,j']>-\infty$, then $A[i',j]>-\infty$ for all $i'\in \{i,i+1,\dots,m\}$ and $j \in \{1,2,\dots, j'\}$. In the definition of convex Monge property, inequality \eqref{eqn:monge} is always considered to hold when its right-hand side evaluates to $-\infty$.

A reverse falling staircase matrix $A$ is convex Monge implies $A$ is \emph{convex totally monotone}: for all $i<i'$ and $j<j'$,
\begin{equation}
    \label{eqn:totalmono}
    A[i,j]< A[i,j'] \Rightarrow A[i',j]< A[i',j'].
\end{equation}

Given an $m\times n$ convex totally monotone matrix, let $j_{\max}(i)$ denote the index of the leftmost column containing the maximum value in row $i$.
Note that condition~\eqref{eqn:totalmono} implies
\begin{equation}
    \label{eqn:rowmaximamono}
    1\le     j_{\max}(1) \le j_{\max}(2) \le \dots \le j_{\max}(m) \le n.
\end{equation}
The SMAWK algorithm \cite{smawk} finds $j_{\max}(i)$ for all $1\le i\le m$.  On tall matrices ($n\ll m$),
its time complexity is near-linear in $n$ (instead of $m$), if we allow a compact output representation  based on \eqref{eqn:rowmaximamono}.
This is formally  summarized in the following theorem.
\begin{theorem}[SMAWK algorithm \cite{smawk}]
    \label{thm:smawk}
    Let an  $m\times n$ convex Monge reverse falling staircase matrix $A$ be implicitly given, so that each entry of $A$ can be accessed in constant time.

    There is a deterministic algorithm that finds all row maxima of $A$ in $O\left (n   (1+\log \left \lceil \frac{m}{n}\right \rceil )\right  )$ time. Its output is compactly represented as $n+1$ integers, $1= r_1\le r_2\le \dots \le r_n \le r_{n+1} = m+1$, indicating that for all $1\le j\le n$ and $r_j\le i <r_{j+1}$,   the leftmost maximum element in row $i$ of $A$ is $A[i,j]$.
\end{theorem}
In the context of knapsack algorithms (e.g., \cite{KellererP04,AxiotisT19,icalp21}), SMAWK algorithm is used to find the $(\max,+)$-convolution $c[i]:=\max_{0\le j\le i}\{a[j]+b[i-j]\}$ between an array $a[0\dd n]$ and a \emph{concave} array $b[0\dd m]$ (that is, $b[i]-b[i-1]\ge b[i+1]-b[i]$). To do this, define matrix $A[i,j] =\begin{cases} a[j]+b[i-j] & j\le i \\ -\infty & j>i\end{cases}$. 
    Note that $A$ is a reverse falling staircase matrix.
One can verify that $A$ is convex Monge: for $i<i',j<j'$ such that  all four terms in \cref{eqn:monge} are finite, we have
\begin{align*}
 &    A[i,j]+A[i',j'] - A[i,j']-A[i',j] \\
     =  \ & b[j-i]+b[j'-i']-b[j'-i]-b[j-i']\\
     \ge \ & 0
\end{align*}
by the concavity of $b$. So SMAWK algorithm can compute all row maxima of $A$, which correspond to the answer of the $(\max,+)$-convolution.

\section{Omitted proofs}
\subsection{Proof of \cref{lem:breakties}}
\label{appendix:breakties}
\breakties*
\begin{proof}
    Suppose instance $I$ has capacity $t$ and $n$ items $(w_1,p_1),\dots,(w_n,p_n)$. Define instance $I'$ with capacity $t$ and items $(w_1,p_1'),\dots,(w_n,p_n')$ with modified profits
    \[ p_i' := (p_i \cdot M + i)\cdot \ww + 1,\]
    where $M:= 1+n+\sum_{i=1}^n i$.
    Then, for any item set $S\subseteq [n]$, 
   we have
   \[ 0 \le \sum_{i\in S}p_i' - M\ww\sum_{i\in S}p_i = |S| + \sum_{i\in S}i\ww   < M\ww,\]
   and hence
\[\sum_{i\in S}p_i  = \left \lfloor \frac{\sum_{i\in S}p_i' }{M\ww}\right \rfloor,\]
    so any optimal solution for $I'$ must also be an optimal solution for $I$.

    For any $i\neq j$, note that $p_i' \bmod (M\ww) = i\ww + 1\neq j\ww + 1 = p_j' \bmod (M\ww)$, so $p_i' \neq p_j'$. 
   If $p_i'/w_i = p_j'/w_j$, then from $p_i'w_j \equiv w_j \pmod{\ww} $ and
   $p_j'w_i \equiv w_i \pmod{\ww} $
    we have $w_j = w_i$, which then contradicts $p_i' \neq p_j'$. So $p_i'/w_i \neq p_j'/w_j$.
\end{proof}

\subsection{Proof of \cref{lem:weightpartition}}
\label{app:partition}
\weightpartition*
\begin{proof}
 Assume the  $n$ items are sorted in decreasing order of efficiency, and the greedy solution is $G=\{1,2,\dots,i^*\}$.  
  Following \cite{chen2023faster}, for each $1\le j\le s$, we define $\ell_j$ to be the smallest $1\le \ell_j\le i^*$ such that $|\supp(\{w_{\ell_j},w_{\ell_j+1},\dots,w_{i^*}\})| \le 2C\sqrt{\ww\log \ww }\cdot 2^j$. Similarly, define $r_j$ to be the largest $i^*<r_j\le n$ such that $|\supp(\{w_{i^*+1},\dots,w_{r_j}\})| \le 2C\sqrt{\ww\log \ww }\cdot 2^j$.
  Note that $[\ell_j,r_j]\subseteq [\ell_{j+1},r_{j+1}]$.
We define the number of layers $s$ to be the smallest $s$ such that $2C\sqrt{\ww\log \ww}\cdot 2^s \ge \ww$. Then, in the last layer $j=s$, we have $\ell_s=1,r_s=n$ by definition. Note that $s < \log_2(\sqrt{\ww})$.

   Then, the partition $\caW= \caW_1 \uplus \caW_2 \uplus \dots \uplus \caW_s$ is defined as follows:
  \begin{itemize}
    \item Let $\caW_{\le j} := \supp(\{w_{\ell_j},w_{\ell_j+1},\dots,w_{r_j}\})$ for all $1\le j\le s$.
        \item Then, let $\caW_1:=\caW_{\le 1}$, and $\caW_j := \caW_{\le j} \setminus  \caW_{\le j-1}$ for $2\le j\le s$.
  \end{itemize}
  Since $[\ell_j,r_j]\subseteq [\ell_{j+1},r_{j+1}]$ and $[\ell_s,r_s]=[1,n]$,
  this construction defines a partition $\caW= \caW_1 \uplus \caW_2 \uplus \dots \uplus \caW_s$. By definition of $\ell_j,r_j$, it automatically satisfies the first property
  $|\caW_j|\le |\caW_{\le j}| \le 4C\sqrt{\ww\log \ww} \cdot 2^{j}$
  from the lemma statement. It remains to verify the second claimed property
  $W( A \cap  I_{{\caW_{>j}}}) \le 4C\ww^{3/2}/2^j$ (the other property $W( B \cap  I_{{\caW_{>j}}}) \le 4C\ww^{3/2}/2^j$ can be proved similarly).
  
 Suppose this property is violated by some 
 optimal exchange solution $(A,B)$  (where $A\subseteq \widebar{G} = \{i^*+1,\dots,n\}, B\subseteq G = \{1,\dots,i^*\}$) and some $1\le j\le s$, i.e., 
 \begin{equation}
     \label{eqn:violate}
 W( A \cap  I_{\caW_{>j}}) > 4C\ww^{3/2}/2^j.
 \end{equation}
In particular, $A \cap  I_{\caW_{>j}} \neq \varnothing$ implies $r_j\neq n$. 
By definition of $r_j$, this means
\begin{equation}
    \label{eqn:largesupp}
    |\supp(\{w_{i^*+1},\dots,w_{r_j}\})| > 2C\sqrt{\ww\log\ww}\cdot 2^j - 1.
\end{equation}
Now consider the following two cases. Recall $\wts(I):= \biguplus_{i\in I}\{w_i\}$ denotes the multiset of item weights in $I\subseteq [n]$.
 \begin{itemize}
    \item Case 
    $|\supp(\wts(A))|\ge C\sqrt{\ww\log \ww}\cdot 2^j$:

Recall from \cref{eqn:sumclose} that $|W(A)-W(B)|< \ww$, so by \cref{eqn:violate} we have
\[ W(B) > W(A)-\ww \ge 4C\ww^{3/2}/2^j - \ww \ge 3C\ww^{3/2}/2^j,\]
where we used $2^j\le 2^s <\sqrt{\ww}$ and assumed $C$ is large enough. Hence, $|\supp(\wts(A))|\cdot W(B) > 3C^2 \ww^2 \sqrt{\log \ww}$, and \cref{lem:proximitysupportsum} (with $N:=\ww$) applied to $\wts(A)$ and $\wts(B)$ implies the existence of two non-empty item subsets $A'\subseteq A, B'\subseteq B$ with the same total weight $W(A')=W(B')$. Then, $(A\setminus A', B\setminus B')$ achieves strictly higher profit than $(A,B)$,  contradicting the optimality of $(A,B)$.
        \item 
        Case 
    $|\supp(\wts(A))|< C\sqrt{\ww\log \ww} \cdot 2^j$:
    
    Then,  we have
    \begin{align*}
&\Big \lvert \supp \big (\wts \big ( \{i^*+1,\dots,r_j\}\setminus A\big )\big ) \Big \rvert    \\
\ge \ & \lvert \supp(\{w_{i^*+1},\dots,w_{r_j}\}) \rvert - \lvert \supp(\wts(A))\rvert \\
 > \ & (2C\sqrt{\ww\log \ww}\cdot 2^j - 1) - C\sqrt{\ww\log \ww}\cdot 2^j \tag{by \cref{eqn:largesupp}}\\
 >\ & C\sqrt{\ww\log \ww}\cdot 2^{j-1}.
    \end{align*}
    Then by \cref{eqn:violate} we have 
    \[\Big \lvert \supp \big (\wts\big (\{i^*+1,\dots,r_j\}\setminus A\big )\big )\Big \rvert \cdot W(A\cap I_{\caW_{>j}}) \ge 2C^2 \ww^2\sqrt{\log \ww}, \]
    and we can apply \cref{lem:proximitysupportsum} to $\wts\big (\{i^*+1,\dots,r_j\}\setminus A\big )$ and $\wts(A\cap I_{\caW_{>j}})$ to obtain non-empty $A'\subseteq A\cap I_{\caW_{>j}}, B'\subseteq\{i^*+1,\dots,r_j\}\setminus A$ with the same total weight $W(A')=W(B')$.
   By definition of $\caW_{\le j}$ we know $\min_{i\in A'} i> r_j \ge \max_{i\in B'}i$.
   Hence, $((A\setminus A') \cup B', B)$ achieves strictly higher profit than $(A,B)$, contradicting the optimality of $(A,B)$.
 \end{itemize}
 We have reached contradiction in both cases, so \cref{eqn:violate} cannot hold. This finishes the proof of the claimed properties of the partitioning $\caW= \caW_1 \uplus \caW_2 \uplus \dots \uplus \caW_s$.
 
Finally we briefly describe how this partitioning  can be computed without actually sorting the $n$ items in $O(n\log n)$ time.  First recall that the break point element $i^*$ of the greedy solution can be found in deterministic $O(n)$ time using linear-time median finding algorithms as in \cite{icalp21}. Then, for all $w\in \caW$, we can find the item in $\{i\in[n]: w_i=w,\, p_i/w_i<p_{i^*}/w_{i^*}\}$ with the highest efficiency, as well as the item in $\{i\in[n]: w_i=w,\, p_i/w_i>p_{i^*}/w_{i^*}\}$ with the lowest efficiency, in $O(n+\ww)$ total time.
We then sort these $2|\caW|$ items by their efficiencies in $O(\ww \log \ww)$ time. Using this information, we can easily compute the weights of the boundary items $\ell_j,r_j$ defined above, and then we can obtain the partitioning.
\end{proof}

\subsection{Proof of \cref{lem:exchange-l0proximity}}
\label{appendix:comb}
Here we show how a slightly weaker version of \cref{lem:exchange-l0proximity} with one extra $\log(2r)$ factor  can be derived from the additive combinatorial result of Bringmann and Wellnitz \cite{BringmannW21} (which built on works of \sarkozy \cite{Sarkozy1,Sarkozy2} and Galil and Margalit \cite{GalilM91}).
In the end we will briefly explain how to remove this logarithmic factor by opening the black box of \cite{BringmannW21}.

\begin{lemma}[Weaker version of \cref{lem:exchange-l0proximity}]
    \label{lem:exchange-l0proximity-weak}
There is a constant $C$ such that the following holds. Suppose two multisets $A,B$ supported on $[N]$ satisfy 
\begin{equation} |\supp_r(A)| \ge C \sqrt{N/r}\cdot \sqrt{\log (2N)}\log (2r) \label{eqn:exchangelemmacond1temp} \end{equation}
for some $r\ge 1$,
and 
\begin{equation}
    \label{eqn:exchangelemmacond2temp}
\Sigma(B) \ge \Sigma(A)-N.
\end{equation}
Then, $\caS^*(A) \cap \caS^*(B) \neq \varnothing$.
\end{lemma}

To prove \cref{lem:exchange-l0proximity-weak}, we need the following lemma by Bringmann and Wellnitz \cite{BringmannW21}.
For a multiset $X$, denote the maximum multiplicity of $X$ by $\mu_X:= \max_{x} \mu_X(x)$.

\begin{lemma}[{\cite[Section 4.1]{BringmannW21}}]
	\label{lem:bw}
	Let $X$ be a non-empty multiset supported on  $[N]$, and 
    \begin{equation}
	 C_{\delta} = 1699200\cdot \log(2|X|)\log^2(2\mu_X),
     \label{eqn:cdelta}
    \end{equation}
    \begin{equation}
	 C_2 =4 + 2\cdot 169920\cdot \log(2\mu_X).
     \label{eqn:c2}
    \end{equation}

	If $|X|^2 \ge C_{\delta} \cdot \mu_X\cdot N$, then there exists an integer \[1\le d\le  4\mu_X \Sigma(X)/|X|^2\]  
	such that for any integer \[t \in [C_2\mu_XN\Sigma(X)/|X|^2,\, \Sigma(X)/2]\]
	that is a multiple of $d$, it holds that $t\in \caS(X)$.
\end{lemma}
\begin{remark}
This statement is paraphrased from	\cite[Theorem 4.3]{BringmannW21} (see also \cite[Definition 3.1]{BringmannW21}). The original statement there stated more generally that $t\in \caS(X)$ if and only if $t\bmod d \in \caS(X) \bmod d$; since $0\in \caS(X)$, the statement here applies in particular to $t$ that is a multiple of $d$. The upper bound $d\le  4\mu_X \Sigma(X)/|X|^2$ stated here can be found in \cite[Theorem 4.1]{BringmannW21}.
\end{remark}

\begin{proof}[Proof of \cref{lem:exchange-l0proximity-weak}]
    The proof uses a similar strategy as \cite{chen2023faster} (and also   \cite{erdos-sarkozy}).
    By the assumption \cref{eqn:exchangelemmacond1temp},
	we can pick $X\subseteq A$ such that 
    \begin{equation}
        \label{eqn:apptemp1}
    \mu_{X}(x) = r \text{ for all } x\in X,
    \end{equation}
    \begin{equation}
        \label{eqn:apptemp2}
    |\supp(X)| \ge C \sqrt{N/r} \cdot \sqrt{\log (2N)}\log (2r),
    \end{equation}
     and \[|X|=r\cdot |\supp(X)| \ge C \sqrt{Nr} \cdot \sqrt{\log (2N)}\log (2r).\]
	Observe that the precondition of \cref{lem:bw}, $|X|^2 \ge C_{\delta}\cdot \mu_X\cdot N$, is satisfied (assuming constant $C$ is large enough), and hence by \cref{lem:bw} there exists an integer 
    \begin{equation}
        \label{eqn:apptemp3}
1\le d\le  4r \Sigma(X)/|X|^2,
    \end{equation}
such that for any integer 
\begin{equation}
    \label{eqn:appinterval}
t \in [C_2rN\Sigma(X)/|X|^2, \Sigma(X)/2]
\end{equation}
	that is a multiple of $d$, it holds that $t\in \caS(X)\subseteq \caS(A)$.

By \cref{eqn:apptemp1} and \cref{eqn:apptemp2}, we have
\[ \Sigma(X) \ge r\cdot \big (1+2+\dots +|\supp(X)|\big ) > r\cdot |\supp(X)|^2 /2 > 100N. \]
	Take the largest subset sum $s\in \caS(B)$ that satisfies $s\le \Sigma(X)/2$ and $s$ is a multiple of $d$, and  let $S\subseteq B$ be the subset achieving $\Sigma(S)=s$. Such $s$ exists, for example by taking $S=\varnothing$. 
    By the assumption \cref{eqn:exchangelemmacond2temp}, 
	 \[ \Sigma(B)-\Sigma(S)\ge (\Sigma(A)-N) - \Sigma(X)/2 \ge (\Sigma(X)-N) - \Sigma(X)/2 > \Sigma(X)/3,\]
     and hence the multiset $B\setminus S$ has size\footnote{For two multisets $S\subseteq B$, the difference $B\setminus S$ is naturally defined by subtracting the multiplicities of elements.}
     \begin{align*}
        |B\setminus S| \ge \frac{\Sigma(B)-\Sigma(S)}{N}
	 & > \frac{\Sigma(X)/3}{N} \\
      & \ge \frac{d|X|^2}{4r}\cdot \frac{1}{3N}\tag{by \cref{eqn:apptemp3}}\\
      & \ge d. \tag{by $|X|^2\ge C_\delta\cdot \mu_X\cdot N$}
     \end{align*}
     So we can pick an arbitrary size-$d$ subset $\{b_1,b_2,\dots,b_d\}\subseteq B\setminus S$, and by the pigeonhole principle, there exist
 $0\le j<j'\le d$ such that $b_1+\dots +b_j \equiv b_1+\dots + b_{j'} \pmod{d}$, and hence $b_{j+1}+\dots +  b_{j'}$ is a multiple of $d$. Denote $S'=\{b_{j+1},\dots,b_{j'}\}\subseteq B\setminus S$, with $|S'|\le d$ and $\Sigma(S') \equiv 0 \pmod{d}$.
Then, the subset sum $\Sigma(S\cup S')\in \caS(B)$ is also a multiple of $d$. By our definition of $s=\Sigma(S)$, we must have $\Sigma(S\cup S')>\Sigma(X)/2$ due to the maximality of $\Sigma(S)$, which means
\begin{equation}
    \label{eqn:apptemp6}
 \Sigma(S) > \Sigma(X)/2 - \Sigma(S') \ge \Sigma(X)/2 - Nd .
\end{equation}

Now we verify that $\Sigma(S)$ belongs to the interval defined in \cref{eqn:appinterval}. The upper bound $\Sigma(S)\le \Sigma(X)/2$ is guaranteed by the definition of $S$. For the lower bound, first note that 
\[Nd \le  N\cdot 4r\Sigma(X)/|X|^2\le N\cdot 4r\Sigma(X) /(C_{\delta}rN)\le 0.1\Sigma(X).\]
Combined with \cref{eqn:apptemp6}, this implies
\[ \Sigma(S)\ge \Sigma(X)/2 - Nd \ge 0.4\Sigma(X).\]
Then, note that the lower bound of the interval \cref{eqn:appinterval} is \[C_2 rN \Sigma(X)/|X|^2\le C_2 rN\Sigma(X)/ (C_{\delta}rN)\le 0.3\Sigma(X).\]
Thus, $\Sigma(S)\in [0.4\Sigma(X),0.5\Sigma(X)]$ belongs to the interval in \cref{eqn:appinterval}. Since $\Sigma(S)$ is a multiple of $d$, this implies $\Sigma(S)\in \caS(X)\subseteq \caS(A)$. Since $\Sigma(S)> 0$, we conclude $\Sigma(S)\in \caS^*(A)\cap \caS^*(B)$.
\end{proof}

We have proved the weaker version of \cref{lem:exchange-l0proximity}. In comparison, the statement of the original \cref{lem:exchange-l0proximity} only requires $|\supp_r(A)| \ge C \sqrt{N/r}\cdot \sqrt{\log (2N)}$ instead of $|\supp_r(A)| \ge C \sqrt{N/r}\cdot \sqrt{\log (2N)}\log(2r)$ (\cref{eqn:exchangelemmacond1temp}). 
Inspecting the proof above, we note that this extra $\log(2r)$ factor  comes from the $\log(2\mu_X)$ factors in \cref{eqn:cdelta,eqn:c2} in the statement of \cref{lem:bw} from \cite{BringmannW21}.   In the proof of \cref{lem:bw} in \cite[Section 4.4]{BringmannW21}, these $\log(2\mu_X)$ factors were incurred in the step that transforms a possibly non-uniform multiset $X$ to a uniform multiset,
where a multiset $X$ is called \emph{uniform} if the multiplicity $\mu_X(x)$ is the same for all elements $x\in X$; see \cite[Lemma 4.28]{BringmannW21}.
When we apply \cref{lem:bw} in our proof, the multiset $X$ is already uniform (\cref{eqn:apptemp1}), so we can avoid this transformation step in \cite{BringmannW21} and thus avoid the extra $\log(2\mu_X)$ factors in \cref{eqn:cdelta,eqn:c2}.
In this way we can prove \cref{lem:exchange-l0proximity} without the extra $\log(2r)$ factor.

\subsection{Proof of \cref{thm:detballs}}
\label{app:detcolorcode}
\detballs*

To prove \cref{thm:detballs}, we use the celebrated deterministic algorithm for finding a coloring of a set system that achieves small discrepancy, using the method of conditional probabilities or pessimistic estimators.

\begin{theorem}[Deterministic set balancing \cite{spencer1987ten,Raghavan88}]
    \label{thm:raghavan}
   Given sets $S_1,S_2,\dots, S_m\subseteq [n]$ where $|S_i|\le b$ for all $i$, there is an $O(n +bm)$-time deterministic algorithm that finds $x\in \{+1,-1\}^n$, such that for every $i\in [m]$, 
   \[ \left \lvert \sum_{j\in S_i}x_j\right \rvert \le 2\sqrt{b\ln(2m)}.\]
\end{theorem}
Note that the algorithm from \cref{thm:raghavan} can run in input-sparsity time, because when coloring an element $j\in [n]$ we only need to update the pessimistic estimators for the sets $S_i$ that contain $j$.
To make this algorithm work in the word-RAM model with $\Theta(\log (n+m))$-bit words, the discrepancy bound is worsened by a constant factor to allow arithmetic operations with relative error $1/\poly(nm)$ when computing the pessimistic estimators.

\cref{thm:detballs} then follows from recursively applying \cref{thm:raghavan}.
\begin{proof}[Proof of \cref{thm:detballs}]
    Without loss of generality we assume $r\ge 2$, and assume $r$ is a power of two, by decreasing $r$ and increasing the size upper bound to $|S_i| \le b_0:= 2r\log_2(2m)$. 
    
    We use the following recursive algorithm with $\log_2(r)$ levels: given size-$b_0$ sets $S_1,\dots,S_m$ from universe $[n]$, use \cref{thm:raghavan} to find  a two-coloring $C_0\colon [n] \to \{+1,-1\}$, and then recurse on two subproblems whose $m$ sets are restricted to the universe $C_0^{-1}(+1)$ and the universe $C_0^{-1}(-1)$ respectively. The discrepancy of $C_0$ achieved by \cref{thm:raghavan} ensures these two subproblems contain sets of size at most $b_0/2 + \sqrt{b_0\ln(2m)}$. We keep recursing in this manner, and the $k$-th level of the recursion tree gives a $2^k$-coloring of the universe $[n]$. Finally the $\log_2(r)$-th level gives the desired $r$-coloring of $[n]$.
     By induction, the maximum size of $S_i$ intersecting any color class is at most $b_{\log_2(r)}$, which is recursively defined as
    \[ b_0 = 2r\log m,\hspace{0.5cm} b_k = b_{k-1}/2 + \sqrt{b_{k-1} \ln(2m)}\hspace{0.2cm} (1\le k\le \log_2(r)).\]
    To solve this recurrence,  first divide by $b_{k-1}$ on both sides and get
\begin{align*}
     \frac{b_k}{b_{k-1}} &\le \frac{1}{2}\exp\left (2 \sqrt{\ln (2m)/b_{k-1}} \right ) \tag{using $1+x\le e^x$}\\
     & \le \frac{1}{2}\exp\Big (2 \sqrt{2^{k-1}\ln (2m)/b_{0}} \Big ) \tag{using $b_j\ge b_{j-1}/2$ inductively}.
\end{align*}
By telescoping,
\begin{align*}
    b_{\log_2(r)} &\le b_0 \prod_{k=1}^{\log_2(r)}\frac{1}{2}\exp\Big (2 \sqrt{2^{k-1}\ln (2m)/b_{0}} \Big ) \\
& = \frac{b_0}{2^{\log_2(r)}}\exp\Big ( 2 \sqrt{\ln (2m)/b_{0}}\sum_{k=1}^{\log_2(r)}\sqrt{2^{k-1}}\Big )\\
& = \tfrac{2r\log_2(2m)}{r}\exp\Big ( 2\sqrt{\tfrac{\ln(2m)}{2r\log_2(2m)}} \tfrac{\sqrt{r} - 1}{\sqrt{2}-1}\Big)\\
& = O(\log m),
\end{align*}
as desired.

    The total time complexity of applying \cref{thm:raghavan} at the $k$-th level of the recursion tree is $O(n + 2^k \cdot b_k m)$.
    From $b_k\ge b_{k-1}/2$, or equivalently $2^k b_k \ge 2^{k-1}b_{k-1}$, we know
     $2^k b_k$ is maximized at $k=\log_2(r)$. Hence, the total time complexity over all $\log_2(r)$ levels is at most 
   \begin{align*}
\sum_{k=1}^{\log_2(r)} O(n +2^k b_k\cdot m) &\le \log_2(r) \cdot O(n + 2^{\log_2(r)}b_{\log_2(r)}\cdot m) \\
     & \le O(n\log r + r m\log r\log m).\qedhere
   \end{align*} 
\end{proof}

\subsection{Proof of \cref{lem:detcolorcode2}}
\label{app:detcolorcode2}
\detcolorsmall*

\begin{proof}

We iteratively apply the following claim:
\begin{claim}
    \label{claim:oneround}
Given $m$ sets $S_1,\dots,S_m \in [n]$ with $|S_i|\le b$, there is a deterministic algorithm in $O(n + mb^3)$ time that computes a coloring $h\colon [n] \to [b^2]$, such that for at least $1/2$ fraction of the $i\in [m]$, $h$ assigns distinct colors to elements of $S_i$.
\end{claim}
To prove \cref{lem:detcolorcode2}, we use \cref{claim:oneround} to find  an coloring $h$, and recurse on the remaining  sets $S_i$ that do not receive distinct colors under $h$. Finally we return all the computed colorings.
Since the number of remaining sets gets halved in each iteration, we terminate within $\log_2 (2m)$ iterations, and the total time complexity is $O(n\log m + mb^3)$.
In the following, it remains to prove \cref{claim:oneround}.

First note that a uniformly random coloring $h\colon [n]\to [b^2]$ assigns distinct colors to a fixed set $S_i$ with probability at least $1 - \binom{|S_i|}{2}\cdot \frac{1}{b^2} \ge 1/2$ by a union bound over all $\binom{|S_i|}{2}$ pairs of distinct $x,y\in S_i$.
We  now use the standard method of pessimistic estimators to derandomize this random construction.
Without loss of generality, assume $|S_i|=b$ for all $i$. We decide the colors $h(j)\in [b^2]$ sequentially for $j=1,2,\dots,n$. Let $x_i^{(j)} = | S_i \cap [j]|$ denote the number of already colored elements in $S_i$. 
Assuming $S_i$ has not received repeated colors among its first $x_i^{(j)}$ elements so far, the probability that coloring the remaining elements causes $S_i$ to have repeated colors is at most $p(x_i^{(j)})$ given by $p(x) :=\frac{1}{b^2}\cdot (x\cdot (b-x) + \binom{b-x}{2})$.

Given the already decided colors $h(1),\dots,h(j)$, define estimator \[e_i^{(j)} = \begin{cases}
   1 &  \text{$S_i$ has already received repeated colors,}\\
   p(x_i^{(j)}) & \text{otherwise.}
\end{cases}\]
Then $\sum_{i=1}^m e_i^{(0)} = m\cdot p(0) <m/2$. 

Now we verify that $\Ex_{h(j+1)\in [b^2]}[e_i^{(j+1)}] \le e_i^{(j)}$. Here it suffices to consider the case where the first $x=x_i^{(j)}$ elements of $S_i$ have not received repeated colors, and $j+1$ is the $(x+1)$-st element in $S_i$ to be colored. Then 
\begin{align*}
    \Ex_{h(j+1)\in [b^2]}[e_i^{(j+1)}] - e_i^{(j)} &= \big (\tfrac{x}{b^2}\cdot 1 +  (1-\tfrac{x}{b^2} )\cdot p(x+1)\big ) - p(x)\\
    & = - \tfrac{x(b+x)(b-x-1)}{2b^4}\\
    & \le 0.
\end{align*}
Hence, $\Ex_{h(j+1)\in [b^2]}[\sum_{i=1}^me_i^{(j+1)}] \le \sum_{i=1}^m e_i^{(j)}$, and by enumerating all $b^2$ colors we can find a color $h(j+1)\in [b^2]$ so that $\sum_{i=1}^me_i^{(j+1)} \le \sum_{i=1}^m e_i^{(j)}$. Eventually we can find a coloring $h\colon  [n] \to [b^2]$ such that $\sum_{i=1}^m e_i^{(n)} \le  \sum_{i=1}^m e_i^{(0)} < m/2$, satisfying the requirement.

Now we analyze the time complexity.
In each iteration $1\le j\le n$, we try all $b^2$ colors and compute the estimators.  Note that the estimators $e_i^{(j)}$ are only updated when an element in $S_i$ is colored. Hence the total time complexity for finding the coloring $h\colon [n] \to [b^2]$ is $O(n\cdot b^2 + b^2\cdot \sum_{i=1}^m |S_i| ) = O(nb^2 + mb^3)$, which can be further reduced to $O(n+mb^3)$ by skipping elements $j\in [n]$ that are not contained in any $S_i$. This finishes the proof of \cref{claim:oneround}.
\end{proof}

\end{document}